\documentclass[a4paper,12pt]{article}
\usepackage[all]{xy}
\usepackage{amssymb}
\usepackage[english]{babel}
\usepackage{mathrsfs}
\usepackage{amsfonts}
\usepackage[centertags]{amsmath}
\usepackage{amsthm}
\usepackage{enumitem}
\usepackage{latexsym,amsmath,amssymb}
\usepackage{graphicx,color}
\usepackage[all]{xy}
\newcommand{\R}{\mathbb{R}}      
\newcommand{\so}{\mathfrak{so}}
\newcommand{\I}{\mathbb{I}}
\newcommand{\FLc}{\mathbb{F}L_c}
\newcommand{\D}{\mathcal{D}}
\newcommand{\E}{\mathcal{E}}
\newcommand{\F}{\mathcal{F}}
\newcommand{\Pp}{\mathcal{P}}
\newcommand{\Teta}{\mathcal{T}_{d\eta}}
\newcommand{\XHc}{X_\mathrm{\small{nh}}}

\setenumerate[1]{label=(\roman*), ref=(\roman*)}
\usepackage[all]{xy}

\numberwithin{equation}{section}
\numberwithin{table}{section}
\numberwithin{figure}{section}

\newtheorem{thm}{Theorem}[section]
\newtheorem{prop}[thm]{Proposition}

\newtheorem{lemma}[thm]{Lemma}

\newtheoremstyle{obs}
{3pt}
{3pt}
{}
{}
{\bfseries}
{.}
{.5em}
{}
\theoremstyle{obs}
\newtheorem{remark}[thm]{Remark}

\def\qed{\ifvmode\removelastskip\fi
	{\unskip\nobreak\hfil\penalty50\hbox{}\nobreak\hfil \hbox{\vrule
			height1.2ex width1.2ex}\parfillskip=0pt \finalhyphendemerits=0
		\par \smallskip}}

\title{}
\author{Garc\'ia-Naranjo,~L.~C.\footnote{Dipartimento di Matematica ``Tullio Levi-Civita", Universit\`a di Padova,
Via Trieste 63, 35121 Padova, Italy. \texttt{luis.garcianaranjo@math.unipd.it}},\, Marrero,~J.~C.\footnote{ULL-CSIC 
Geometr\'ia Diferencial y Mec\'anica Geom\'etrica, Departamento de Matem\'aticas, Estad\'istica e 
Investigaci\'on Operativa and Instituto de Matem\'aticas y Aplicaciones (IMAULL), Facultad de Ciencias 38071, Universidad de La Laguna, Tenerife, Spain. \texttt{jcmarrer@ull.edu.es}}, \\Mart\'in de Diego,~D.\footnote{Instituto de Ciencias Matem\'aticas (CSIC-UAM-UC3M-UCM), Campus de Cantoblanco, UAM 
C/ Nicolas Cabrera, 15, 
 28049 Madrid (SPAIN). \texttt{david.martin@icmat.es}}\, and Petit Vald\'es,~P.~E.\footnote{Dipartimento di Matematica, Universit\`a di Trento, Via Sommarive 14, Povo-Trento, 38123, Italy. \texttt{pe.petit@unitn.it}}}

\begin{document}

	\title{ Almost-Poisson brackets for nonholonomic systems with gyroscopic terms and Hamiltonisation}

	\maketitle

	\abstract{We extend known constructions of almost-Poisson brackets and their gauge transformations
	to nonholonomic systems whose Lagrangian is not mechanical but possesses a gyroscopic term  linear
	in the velocities. The new 
	feature introduced by such a term 
	is that the Legendre transformation is an affine, instead of linear, bundle isomorphism between the tangent and cotangent
	bundles of the configuration space and  some care is needed  in the development of the geometric formalism. 
	 At the end of the day, the affine nature of the Legendre transform  is reflected in the 
	affine dependence of the brackets that we construct on the momentum variables. Our study is motivated
	by  a wide class of nonholonomic systems involving rigid bodies with internal rotors which are of interest in control.
	Our construction  provides a natural geometric framework for the (known) Hamiltonisations
	of the Suslov and Chaplygin sphere problems with a gyrostat.   }

\vspace{0.5cm}


\section{Introduction}

The geometric study of nonholonomic systems has been an active subject of research in the last decades.
Some early influential references are \cite{Weber2,Koiller92,BatesSniaticky, BKMM, CaLeMaMa, CaCoLeMa}. We refer
the reader to  the books 
 \cite{Neimark,Cortes,Cu-Du,Bloch} for a broad exposition and more extensive bibliography.
 It is now widely accepted that the underlying 
geometric structure of these systems  is an almost-Poisson bracket which, as a consequence
of the non-integrability of the nonholonomic constraints, fails to satisfy the Jacobi identity. The origin of this theory can be traced 
back to a 1951 paper by Eden \cite{Eden51} (see also \cite{LeLaLoMa2023} and references therein). In his paper, Eden presents a local definition of the nonholonomic bracket for observables in the constraint submanifold. For an intrinsic definition and basic geometric properties  of the bracket, see \cite{MaschkevanderSchaft,KoonMarsden, ILMM99,CanLeoMar,GLMM,LMD2010,LeLaLoMa2023}. 
Despite the failure of the Jacobi identity, the almost-Poisson formulation has been useful to investigate
existence of invariant measures \cite{FeGaMa} and has received wide interest in the context of reduction 
\cite{CaLeMaMa,KoonMarsden,GN07,CdLMM,Cu-Du,GN10}, and more 
specifically
in connection with {\em Hamiltonisation} where one attempts to find a true Hamiltonian structure for the reduced dynamics 
\cite{FedoJova,Ramos,Ehlers,BorisovMamaev,Jova,GN10,Bolsinov2015,BaGa,BalFer,Naranjo,BalYap}.


In this paper we focus on the almost-Poisson structure of nonholonomic systems for which the Lagrangian $L: TQ\rightarrow \R$ 
contains a linear term in the velocities in addition to the usual terms coming from the kinetic and potential energies.  In 
bundle coordinates $(q,\dot q)$ for the tangent bundle $TQ$ of the $n$-dimensional configuration manifold $Q$, the 
expression for $L$  is then given by\footnote{We employ the standard sum convention  over repeated indices.} 
\begin{equation}
\label{eq:Lag-intro}
L(q,\dot q)=\frac{1}{2}\tilde g_{ij}(q)\dot q^i\dot q^j + \tilde \eta_j(q)\dot q^j-V(q).
\end{equation}
The first term is the kinetic energy which  is defined  in terms of a
Riemannian metric $g:=\tilde g_{ij}(q)\, dq^i\otimes dq^j$ on $Q$. The third term is the 
 potential energy given by a smooth function $V:Q\to \R$. The linear term in the velocities, which is the
novel aspect in this work, is sometimes referred to as a 
{\em magnetic} or {\em gyroscopic} term, and is determined by a 1-form $\eta:= \tilde \eta_j(q) \, dq^j$. In the absence of the gyroscopic
term, it is common to say that one is working with a {\em mechanical Lagrangian}.

We assume that the mechanical system with Lagrangian  \eqref{eq:Lag-intro} is under the influence of $n-r$ nonholonomic linear constraints
\begin{equation}
\label{eq:const-intro}
A_j^\alpha(q)\dot q^j=0,  \qquad \alpha=r+1,\dots, n,
\end{equation}
which are  independent and determine a regular non-integrable distribution $\D$ of
rank $r\geq 2$ on $Q$.

As for unconstrained systems,  the gyroscopic term in the Lagrangian 
manifests itself in the equations of motion only through the exterior differential $d\eta$. 
Hence, if $\eta$ is a closed 1-form, the magnetic term  has absolutely no influence in the dynamics and can be disregarded. 
On the other hand, if  $\eta$ is non-closed, it will typically destroy the reversibility of the dynamics.

%
%

The geometry of the  systems described above has been greatly overlooked with most references usually
focusing on the case of mechanical Lagrangians. An  exception  is the recent work by Dragovi\'c, Gaji\'c and Jovanovi\'c \cite{DragovicGajicJova2023},
which concerns the study of these systems in the presence of additional symmetries of a very specific type (the so-called Chaplygin 
systems) giving attention to their Hamiltonisation.   

We begin by motivating our study with  a brief subsection which explains how 
 gyroscopic nonholonomic systems arise  through a mechanism which is very similar to the classical  Routh reduction method 
for unconstrained systems. A second motivation (for which we give no details)
 is the development of an almost-Poisson bracket formulation of a remarkable 
class of affine nonholonomic systems possessing an energy-like integral \cite{FassoSanso,FaGaSa,BMB2015} that we plan to develop in a separate
paper \cite{Futurepaper}.


\subsection{Motivation: nonholonomic Routh reduction}
\label{SS:Routh-intro}

Consider a nonholonomic system of mechanical type on a configuration manifold $M=Q\times K$, which is the product
of the $n$-dimensional  manifold $Q$ and an abelian Lie group $K$. To fix ideas, suppose  that $K$ is the $l$-torus
$\mathbb{T}^l=(S^1)^l$. Let  $(q^j,\theta^J)\in Q\times K$  be  coordinates on $M$ and assume that  $\theta^J$ are cyclic variables of 
   a mechanical Lagrangian $\mathcal{L}:TM\to \R$, namely,
\begin{equation*}
\mathcal{L}=\mathcal{L}(q,\dot q,\dot \theta )= \frac{1}{2}\mathcal{G}_{ij}(q)\dot q^i\dot q^j
+ \mathcal{G}_{iJ}(q)\dot q^i\dot \theta^J+\frac{1}{2}\mathcal{G}_{IJ}(q)\dot \theta^I\dot \theta^J - \tilde V(q).
\end{equation*}
Suppose, in addition, 
that the nonholonomic constraints are independent of $\theta$ and $\dot \theta$ and  are given in terms of
$q$ and $\dot q$ by \eqref{eq:const-intro}. 

The above setup is encountered, for instance,
 when rigid bodies, containing   $l$ internal rotors, are subjected to nonholonomic constraints. In such
case, the coordinates $q^j$ describe  the position and orientation of the rigid bodies whereas the angles 
$\theta^J$ determine the 
orientation of the rotors (relative to axes that are fixed in the corresponding body frames). 
The independence of the nonholonomic constraints \eqref{eq:const-intro} on   $\theta$ and $\dot \theta$
corresponds to our assumption that the rotors lie  inside of the bodies and their motion is not directly affected by the constraints.
Some concrete examples of this framework are the snakeboard, the unicycle with rotor and the Chaplygin sleigh with
rotor (see e.g. \cite{Bloch}).  Other examples involve bodies containing a gyrostat which roll without slipping 
 on surfaces, or the gyrostat generalisations of the classical Suslov and
Veselova problems.

In accordance with the Lagrange-D'Alembert principle, the equations of motion of the system  take the form
\begin{equation}
\label{eq:motion-intro-Routh}
\begin{split}
\frac{d}{dt}\left ( \frac{\partial \mathcal{L}}{\partial \dot q^j }\right ) -  \frac{\partial \mathcal{L}}{\partial q^j } &= \lambda_\alpha A^\alpha_j(q), \qquad   j=1,\dots, n, \\ 
\frac{d}{dt}\left ( \frac{\partial \mathcal{L}}{\partial \dot \theta^J }\right ) &=0, \qquad   J=1,\dots, l,
\end{split}
\end{equation}
where the multipliers $\lambda_\alpha$  are uniquely determined by  the constraints  \eqref{eq:const-intro}. It is therefore clear that
the momenta
\begin{equation*} 
p_J :=  \frac{\partial \mathcal{L}}{\partial \dot \theta^J }= \mathcal{G}_{iJ}(q)\dot q^i+\mathcal{G}_{IJ}(q)\dot \theta^I, \qquad
J=1,\dots, l,
\end{equation*}
are first integrals.  Let us consider the restriction of the system  \eqref{eq:motion-intro-Routh} to a level set of these integrals by
setting $p_J=\mu_J$ for some fixed $\mu=(\mu_1,\dots, \mu_l)\in \R^l$. Along such level set we may express\footnote{As usual,
we denote by $\mathcal{G}^{IJ}$  the entries of the inverse matrix of the block $\mathcal{G}_{IJ}$.}
\begin{equation}
\label{eq:dot-theta-Routh}
\dot \theta^J=\mathcal{G}^{IJ}(q)(\mu_I-\mathcal{G}_{iI}(q)\dot q^i),\qquad J=1,\dots, l,
\end{equation}
and we may therefore eliminate $\dot \theta$ from the first set of equations in \eqref{eq:motion-intro-Routh} to obtain a reduced system
involving $q$ and $\dot q$ only. As is well known (see e.g. \cite{MaRa1}), such elimination is conveniently done in terms of the
classical Routhian function $R^\mu=R^\mu(q,\dot q)$ defined by 
\begin{equation*}
R^\mu(q,\dot q):=\left [ \mathcal{L}(q,\dot q,\dot \theta ) - \mu_J\dot \theta ^J \right ]_{p_J=\mu_J},
\end{equation*}
with the convention that in the right hand side  $\dot \theta$ is written in terms of $(q,\dot q)$ as in \eqref{eq:dot-theta-Routh}.
The remarkable property  of the Routhian  is that 
\begin{equation*}
  \frac{d}{dt}\left ( \frac{\partial R^\mu }{\partial \dot q^j }\right ) -  \frac{\partial R^\mu}{\partial q^j } = \left [ \frac{d}{dt}\left ( \frac{\partial \mathcal{L}}{\partial \dot q^j }\right ) -  \frac{\partial \mathcal{L}}{\partial q^j } \right ]_{p_J=\mu_J}, \qquad   j=1,\dots, n,
\end{equation*}
where we again think of $\dot \theta=\dot \theta(q,\dot q)$ in the right hand side.
Therefore, if we write $L(q,\dot q):=R^\mu(q,\dot q)$, the reduced system can be written as
\begin{equation}
\label{eq:motion-intro-Routh-reduced}
\begin{split}
\frac{d}{dt}\left ( \frac{\partial L }{\partial \dot q^j }\right ) -  \frac{\partial L}{\partial q^j } &= \lambda_\alpha A^\alpha_j(q), \qquad   j=1,\dots, n,
\end{split}
\end{equation}
together with the nonholonomic constraints   \eqref{eq:const-intro}.

Now note that  \eqref{eq:motion-intro-Routh-reduced}  can be interpreted as the 
equations of motion of a nonholonomic system on the configuration manifold $Q$, with constraints given by  \eqref{eq:const-intro} and
with Lagrangian $L=R^\mu:TQ\to \R$. Our point is that $L$ will {\em in general have a gyroscopic term}. Indeed, as shown in e.g. Section 8.9 in \cite{MaRa1}, 
$L=R^\mu$ is given by \eqref{eq:Lag-intro} with
\begin{equation}
\label{eq:Routhian-terms}
\begin{split}
\tilde g_{ij}=\mathcal{G}_{ij} - \mathcal{G}_{Ii}\mathcal{G}^{IJ}\mathcal{G}_{Jj}, \qquad  \tilde \eta_j=\mathcal{G}_{Ij}\mathcal{G}^{IJ}\mu_J, \qquad
V= \tilde V + \frac{1}{2}\mathcal{G}^{IJ}\mu_I\mu_J.
\end{split}
\end{equation}

The class of examples presented above is especially relevant in the context of control
theory  since it is natural to add controls to the rotors to stabilise or accomplish desired motions of the bodies, see e.g. \cite{WaKr,BoKiMa,GaBa,Ohsawa,GaTo} and references therein.



\subsection{The almost-Poisson formulation}
\label{SS:Almost-Poisson-intro}

We  now present a brief  coordinate description  of how the equations of motion
of the mechanical system with configuration space $Q$, Lagrangian $L$ given by \eqref{eq:Lag-intro}, 
and nonholonomic constraints \eqref{eq:const-intro},
can be formulated in terms of almost-Poisson brackets.

Our approach relies on the choice of vector frames which define quasi-velocities that are conveniently adapted to 
the constraints. The usefulness of this procedure has been by now appreciated by different researchers in nonholonomic mechanics
\cite{Koiller92,Ehlers,CdLMM,BoMaZe,GLMM,LMD2010,BoKiMa2}.   

Let $\{e_a\}$ be a  basis of sections of the distribution
$\D$ defined by the constraints  \eqref{eq:const-intro}, that is, $e_a\in{\mathfrak X}(Q)$ with $e_a(q)\in \D_q$ for all $q\in Q$, and let $\{e_\alpha\}$ be a basis of 
sections of $\D^\perp$ where
$\perp$ means perpendicular with respect to the kinetic energy metric.\footnote{The index convention followed throughout this work is that latin indices $a,b,c$   run from $1$ to $r=\mbox{rank}(\D)$
whereas  greek indices $\alpha,\beta, \gamma$  run from $r+1$ to $n$. The indices $i,j,k$ run from the whole range $1$ to $n=\dim(Q)$.}
Such bases may, of course, only exist locally. A tangent vector $v_q\in T_qQ$ is written uniquely as a linear combination
\begin{equation*}
v_q =v^a e_a(q) +v^\alpha e_\alpha(q),
\end{equation*}
and hence $(q^i,v^a,v^{\alpha})$ can be used as coordinates for $TQ$. These coefficients $v^a$ and $v^\alpha$ are termed
{\em quasi-velocities}.
The constraint space $\D$ may then be interpreted as the submanifold of $TQ$ determined 
by the conditions $v^{\alpha}=0$ and has
 coordinates $(q^i,v^a)$. As will be explained in Subsection \ref{sec:lagra-constrained}, the equations of motion \eqref{eq:motion-intro-Routh-reduced} are equivalent
 to the following  equations for the evolution of  $(q^i(t),v^a(t))$, which do not involve the multipliers $\lambda_\alpha$: 
\begin{equation}
\label{eq:motion-intro}
\begin{split}
\\ & \frac{d}{dt}\left ( \frac{\partial L_c}{\partial v^a}  \right )-\rho^i_a(q) \frac{\partial L_c}{\partial q^i}+C_{ab}^d(q)v^b \frac{\partial L_c}{\partial v^d}=-C_{ab}^\alpha(q)v^b \eta_\alpha(q), \qquad a=1,\dots, r, \\
 &\frac{d q^i}{dt} = \rho^i_a(q)v^a, \qquad i=1,\dots, n.
\end{split}
\end{equation}
In the above equations $L_c:=\left . L \right |_\D :\D\to \R$ is the {\em constrained Lagrangian} expressed as a function of the  coordinates $(q^i,v^a)$ for $\D$. The remaining coefficients appearing in the equations
are determined by the conditions
\begin{equation*}
\begin{split}
& e_a(q)=\rho^i_a(q) \frac{\partial}{\partial q^i}, \qquad a=1,\dots, r, \\
& e_\alpha(q)=\rho^i_\alpha(q) \frac{\partial}{\partial q^i}, \qquad \alpha=r+1,\dots, n, \\
&\eta_\alpha(q):=\rho^i_\alpha(q)\tilde \eta_i(q), \qquad \alpha=r+1,\dots, n, \\
& [e_a,e_b] (q)=C_{ab}^j(q)e_j(q)= C_{ab}^d(q)e_d(q)+ C_{ab}^\alpha(q)e_\alpha(q), \qquad a,b=1,\dots, r, 
 \end{split}
\end{equation*}
where $ [\cdot , \cdot]$ is the Lie bracket of vector fields. An important observation is
 that the skew-symmetry of $ [\cdot , \cdot]$ implies 
that $C_{ab}^j(q)$ is antisymmetric with respect to the lower indices, i.e. $C_{ab}^j(q)=-C_{ba}^j(q)$.

To obtain an almost-Poisson formulation of the equations of motion \eqref{eq:motion-intro} we consider the
 {\em constrained Legendre 
transformation} and {\em constrained Hamiltonian}
\begin{equation*}
p_a:=\frac{\partial L_c}{\partial v^a}, \qquad H_c:=H_c(q^i,p_a)=p_av^a-L_c.
\end{equation*}
As usual, the regularity of the Lagrangian $L$ guarantees that the quasi-velocities $v^a$ may be uniquely expressed 
in terms of $(q^i,p_b)$ which explains why we can use $(q^i,p_a)$ as coordinates of our phase space and why
 the constrained Hamiltonian $H_c$ may be considered as a function of $(q^i,p_a)$. 
The standard chain rule calculations of analytical mechanics show that 
\begin{equation*}
 \frac{\partial H_c}{\partial q^i}=- \frac{\partial L_c}{\partial q^i}, \qquad  \frac{\partial H_c}{\partial p_a}=v^a,
\end{equation*}
and, therefore, the equations of motion \eqref{eq:motion-intro} can be written as the following system of 
equations for the evolution of  $(q^i(t),p_a(t))$
\begin{equation}
\label{eq:motion-intro-Ham}
\begin{split}
 &\frac{d q^i}{dt} = \rho^i_a(q)\frac{\partial H_c}{\partial p_a}, \qquad i=1,\dots, n,\\
 &\frac{d p_a}{dt} = -\rho^i_a(q) \frac{\partial H_c}{\partial q^i}-C_{ab}^d(q)p_d \frac{\partial H_c}{\partial p_b} -C_{ab}^\alpha(q) \eta_\alpha(q)\frac{\partial H_c}{\partial p_b}, \qquad a=1,\dots, r.
\end{split}
\end{equation}
Given that the  partial derivatives of the constrained  Hamiltonian $H_c$ appear linearly in the right hand side, we may
write \eqref{eq:motion-intro-Ham} in matrix form:\footnote{When we use the abbreviation $(q,p)$ one should keep in 
mind that the index $i$ of $q$ runs from $1$ to $n$ whereas the range of the index $a$ of $p$ is only $1\leq a \leq r=\mbox{rank}(\D)$.}
\begin{equation}
\label{eq:Pi-matrix}
\frac{d}{dt}
\begin{pmatrix}
 q \\\vspace{-0.3cm}  \\ p 
\end{pmatrix}
= \Pi(q,p)\begin{pmatrix}
\frac{\partial H_c}{\partial q}  \\\vspace{-0.3cm}  \\\frac{\partial H_c}{\partial p}  
\end{pmatrix}, \qquad   \Pi(q,p):=\begin{pmatrix}
0_{n\times n} & \rho(q)  \\  -\rho(q)^T & \mathcal{C}(q,p)+\mathcal{A}(q)  
\end{pmatrix},
\end{equation}
 where the block matrix $ \rho(q)$ is $n\times r$ and has entries  $ \rho(q)_{ia}= \rho_a^{i}(q)$ and $\mathcal{C}(q,p)$ and 
 $\mathcal{A}(q)$ are antisymmetric $r\times r$ block matrices with entries
 \begin{equation*}
\mathcal{C}(q,p)_{ab}=-C_{ab}^d(q)p_d, \qquad \mathcal{A}(q)_{ab}=-C_{ab}^\alpha(q) \eta_\alpha(q).
\end{equation*}

Now define the {\em  nonholonomic bracket} of functions $F_1, F_2$ of the variables $(q,p)$ by
\begin{equation}
\label{eq:defBrack-Intro}
\{F_1,F_2\}(q,p):=\displaystyle\begin{pmatrix}
\frac{\partial F_1}{\partial q}^T  & \frac{\partial F_1}{\partial p}^T  \end{pmatrix} \Pi(q,p)\begin{pmatrix}
\frac{\partial F_2}{\partial q}  \\\vspace{-0.3cm} \\\frac{\partial F_2}{\partial p}  
\end{pmatrix}.
\end{equation}
In particular, the brackets of the coordinate functions are
\begin{equation*}
\{q^i,q^j\}=0,\qquad \{q^i,p_a\}=\rho^i_a, \qquad \{p_a,p_b\}=-C_{ab}^dp_d-C_{ab}^\alpha \eta_\alpha.
\end{equation*}

It is clear that the nonholonomic bracket  is affine in the variables $p_a$ and satisfies Leibniz rule. It is also skew-symmetric since the 
matrix  $ \Pi(q,p)$ is antisymmetric. 
We may therefore write \eqref{eq:motion-intro-Ham} in
{\em almost-Hamiltonian} form:
\begin{equation*}
\label{eq:almost-Ham-Intro}
\frac{dq^i}{dt}=\{q^i,H_c\}, \quad i=1,\dots, n, \qquad \frac{dp_a}{dt}=\{p_a,H_c\}, \quad a=1,\dots, r.
\end{equation*}

An interesting  observation, which is useful in the treatment of examples and we develop in this work, 
 is that, for the purposes of describing the dynamics,  the block matrix $\Pi(q,p)$ 
in \eqref{eq:Pi-matrix} can be instead taken as  
\begin{equation}
\label{eq:IntroPiB}
\Pi^\mathcal{B}(q,p):=\begin{pmatrix}
0_{n\times n} & \rho(q)  \\  -\rho(q)^T & \mathcal{C}(q,p)+\mathcal{A}(q) +   \mathcal{B}(q,p)
\end{pmatrix},
\end{equation}
 as long as the (antisymmetric $r\times r$) block $\mathcal{B}(q,p)$ satisfies 
 \begin{equation}
 \label{eq:Bgradhiszero}
 \mathcal{B}(q,p)\frac{\partial H_c}{\partial p}=0.
\end{equation}
In order to preserve the structure,  suppose that the block $\mathcal{B}(q,p)$ 
has an affine  dependence in $p$ (which is something that  will always hold in our construction, see Remark \ref{rmk:affinegauge}).
The point is that a modified nonholonomic  bracket of functions may be defined replacing
$\Pi(q,p)$ by $\Pi^\mathcal{B}(q,p)$ in   \eqref{eq:defBrack-Intro}, and such modified bracket will also 
describe the dynamics and will retain the geometric properties. This kind of bracket modifications go back to \cite{LGN-thesis, GN10} and were 
coined {\em gauge transformations}
in \cite{BaGa}. For mechanical Lagrangians, and under appropriate hypothesis, gauge transformations are known to be necessary to produce Casimirs  \cite{GNMo18} and are essential to study Hamiltonisation in the Poisson framework \cite{GN10,BaGa,Bal2014,Bal2017,BalYap}.

\subsection{Summary of results}

This work gives intrinsic geometric constructions of the nonholonomic  bracket 
and  its gauge transformations described above and clarifies their main geometric properties. In our approach, the  
brackets are defined on the dual bundle $\D^*$ which we identify with the annihilator $(\D^\perp)^\circ\subseteq T^*Q$ and interpret as the phase space
manifold for the almost-Hamiltonian description of the dynamics. The variables $p_a$, introduced in  Subsection \ref{SS:Almost-Poisson-intro},
 are linear coordinates on 
the fibres of $\D^*$, whereas, as indicated above,
 $q^i$ are coordinates on the base configuration manifold $Q$. The new feature
introduced by the gyroscopic 1-form $\eta$ is that the Legendre transformation is an {\em affine}, instead of  
 linear, bundle map between $TQ$ and $T^*Q$. This property ultimately results in the  affine nature of the brackets which
can be appreciated in the  dependence on $p$ of the bottom right block of the matrix $\Pi(q,p)$ in \eqref{eq:Pi-matrix},
and similarly for the matrix $\Pi^\mathcal{B}(q,p)$. 

The nonholonomic   bracket that we build  is in fact a particular instance of the one given by Maschke and van der Schaft \cite{MaschkevanderSchaft}
 and whose intrinsic nature was identified by Marle \cite{Marle98} and further clarified by 
  Ibort et al \cite{ILMM99} (see also \cite{CanLeoMar,LeLaLoMa2023}). The novelty of our work with 
respect to \cite{Marle98,CanLeoMar,ILMM99} is to give a more elementary intrinsic construction 
that involves only metric projections (with respect to the kinetic energy metric) on subbundles of $TQ$  
instead of the more sophisticated symplectic projections on subbundles of 
$T(T^*Q)$ used in  \cite{Marle98,CanLeoMar,ILMM99}. This approach had been followed before in the case of mechanical Lagrangians \cite{GLMM,LMD2010} (see also \cite{LeLaLoMa2023}) but, 
to the best of our knowledge,  is missing for Lagrangians with gyroscopic terms. On the other hand, the
development of gauge transformations in our setting seems to be completely 
 new since previous constructions  \cite{BaGa,Bal2014,Bal2017,GNMo18,BalYap}
  focus
 exclusively
in the case of mechanical Lagrangians. Our definition  of 
gauge transformations is inspired by \cite{GNMo18} where the bracket modification  is constructed in terms of 
a   $3$-form $\Lambda$ on $Q$. Among other things, this approach provides  a convenient way to ensure that the condition \eqref{eq:Bgradhiszero} 
 is satisfied  by exploiting the alternating skew-symmetry of $\Lambda$.

Our work also considers the  reduction  in the presence of symmetries in the context of the brackets that we construct.
More precisely, we assume that  that a  Lie group $G$ acts freely and properly on  $Q$, and its tangent lift preserves $L$ and $\D$.
Under this hypothesis, both the dynamics and the brackets drop to the quotient space  $\D^*/G$ 
and the reduced dynamics can be described in terms of the reduced brackets (an additional  invariance condition
on the underlying 3-form $\Lambda$ is needed for the reduction of the gauged transformed brackets). The interest of this
procedure is that, as occurs in certain examples, a particular choice of $\Lambda$ may lead to a reduced
bracket which  actually satisfies the Jacobi identity (perhaps after  multiplication by a conformal factor), thereby giving
a {\em true} Hamiltonian description of the reduced dynamics.  As mentioned above, this
phenomenon is known as Hamiltonisation. 

We exemplify our construction by carefully treating the Suslov and Chaplygin sphere problems with a gyrostat within our framework.
We show  that Hamiltonisation is achieved in both examples by proceeding as above.  Even though the 
Hamiltonisation of these systems
was known (see \cite{BorMam2008} for the Chaplygin sphere), our approach clarifies 
their origin within a general reduction theory. 
Of both examples, the treatment of the Chaplygin sphere with gyrostat is more interesting for two reasons. On the one hand, 
 a gauge transformation 
is needed to make a certain first integral, which is  affine in the velocities, into a Casimir of the reduced bracket.
This  suggests that some constructions in \cite{GNMo18} admit a generalisation to case of gyroscopic Lagrangians.
We do not pursue this generalisation here since it requires a development of the theory of first integrals of nonholonomic systems
 with gyroscopic Lagrangians which is beyond our scope. On the other hand, 
the  reduced bracket  which Hamiltonises the problem has  rank 4. Therefore, in contrast with the scenario  encountered for rank 2 Poisson  brackets \cite{FassoGiaSan}, it is unlikely that its existence can be proved using purely dynamical arguments,
and a geometric approach like ours seems unavoidable.

\subsection{Structure of the paper}

We begin by presenting the equations of motion  in both 
 the Lagrangian (Section \ref{sec:Lag}) and Hamiltonian (Section \ref{sec:Ham}) formalisms.
  In both cases, we start by reviewing  the unconstrained case to provide the groundwork
  for the presentation, and give intrinsic and coordinate descriptions of the vector fields which describe the 
  nonholonomic dynamics. Moreover, in Section \ref{sec:Ham} we clarify 
  the geometric relevance of the gyroscopic term in the Legendre transformation.  Our main theoretical contributions
  are given in Section \ref{sec-almost} where we present the intrinsic construction of the nonholonomic bracket,
  its gauge transformations, and discuss their main properties.
  We then consider symmetries and reduction in Section \ref{s:Symmetry}.   Our presentation in this section is 
   somewhat brief and synthetic since the presence of the gyroscopic $1$-form $\eta$ does not change
  things significantly and   such considerations
 are abundant in the literature for mechanical Lagrangians   (see e.g. Chapter 3 in \cite{Cu-Du}). The examples 
 are treated in Sections \ref{S:examples-description}
 and \ref{S:examples-bracks}. We first introduce the systems and give their reduced equations of motion and (known)
 Hamiltonisations in Section \ref{S:examples-description}, and then prove  that such Hamiltonisations
  may be understood within our geometric framework in Section \ref{S:examples-bracks}.
 
We finally mention that some of the material in this paper was part of the Phd thesis of the fourth author \cite{Petit-thesis} and some
overlaps  may be present.

\section{Lagrangian formalism}
\label{sec:Lag}

We begin our discussion from the Lagrangian viewpoint of mechanics on which the
D'Alembert's principle for the reaction forces is more easily formulated. 
Let $Q$ be an $n$-dimensional  configuration manifold. 
We consider the Lagrangian  $L: TQ\rightarrow {\mathbb R}$ defined in terms of the following three ingredients:
%
%
%
\begin{itemize}
\item a  Riemannian metric $g$  on  $Q$, which defines the kinetic energy ${\mathcal K}_g: TQ \to {\mathbb R}$ on $TQ$ 
given by ${\mathcal K}_g(v_q)=\frac{1}{2} g_q(v_q, v_q)$;
%
\item  a  $1$-form $\eta$  on $Q$, that we'll refer to as the {\em gyroscopic 1-form}, which defines a  fiberwise linear function, $\eta^{\ell}: TQ\rightarrow {\mathbb R}$  by
$\eta^{\ell}(v_q)=\langle \eta(q), v_q\rangle$ where $v_q\in T_qQ$;
\item a potential energy function $V\in C^{\infty}(Q)$.
\end{itemize}
The Lagrangian  $L: TQ\rightarrow {\mathbb R}$ is then given by
\begin{equation}\label{def:lagrangian}
L={\mathcal K}_g + \eta^{\ell} - V\circ \tau_Q,
\end{equation}
where $\tau_Q: TQ\rightarrow Q$ is the tangent bundle projection. We will write $L=L_{mech}+ \eta^{\ell}$ where $L_{mech}:={\mathcal K}_g - V\circ \tau_Q$ is the {\em mechanical Lagrangian}.

\subsection{The free case}

We rapidly recall  aspects of the Lagrangian system on $TQ$ in the absence of constraints to introduce the necessary notation.

%
%
\subsubsection*{Local expressions}
Let  $(q^i)$ be local coordinates on  an open subset $U$ of $Q$. The kinetic energy $g$ and 1-form $\eta$ have local expressions
$g=\tilde g_{ij}(q)\, dq^i\otimes dq^j$, and $\eta=\tilde \eta_i(q)\, dq^i$, and thus
\[
L(q^i, \dot q^i)=\frac{1}{2}\tilde g_{ij}(q)\dot q^i\dot q^j+\tilde \eta_i(q)\dot q^i-V(q).
\]
The equations of motion are the well-known Euler-Lagrange equations 
\begin{equation}
\label{eq:EL-bundle}
\dot q^i=\frac{dq^i}{dt}, \qquad \frac{d}{dt}\left ( \frac{\partial L}{\partial \dot q^i} \right ) - \frac{\partial L}{\partial q^i}=0, \quad i=	1,\dots, n.
\end{equation}

Looking ahead at the nonholonomic case, we give local expressions for $L$ and the equations of motion
 in terms of  quasivelocities 
which are coordinates for the fibers of $TQ$ relative to a configuration-dependent frame.
 Let $\{e_i\}$ be such
 local frame; that is,  $e_i$, $i=1,\dots, n$, are
linearly independent vector fields in the neighbourhood $U$.
Any tangent vector $v_q\in T_qQ$, with $q\in U$ may be uniquely written as $v_q=v^ie_i(q)$, for certain {\em quasivelocity}
scalars $v^i$ . 
Then $(q^i,v^i)$ are local coordinates for $TQ$ which generalise the
 bundle coordinates  $(q^i,\dot q^i)$.
The frame vector fields $e_i$ may be uniquely written as 
\begin{equation}
\label{eq:rhodef}
e_i(q)=\rho^j_i(q)\frac{\partial}{\partial q^j}, \qquad 1\leq i\leq n, \quad \forall q\in U,
\end{equation}
where $\rho^j_i\in C^\infty (U)$ are the entries of an invertible matrix $(\rho^j_i(q))$.
  Associated to the  frame $\{e_i\}$ there is a  dual co-frame $\{e^i\}$ determined by the conditions $e^{i} \in \Omega^1(U)$ and $\langle e^i, e_j\rangle=\delta^i_j$. Then $dq^i=\rho^i_k(q)e^k(q)$, and the kinetic energy metric $g$ and the 1-form $\eta$ write as
\[
g=g_{ij}\, e^i\otimes e^j\; ,\qquad \eta=\eta_i  \, e^i,
\]
where we have omitted the $q$ dependence (as we will often do below)  and we have summarised
\begin{equation*}
g_{ij}:=\rho_i^k\rho_j^l \tilde g_{kl}, \qquad \eta_i:=\rho_i^k\tilde \eta_k.
\end{equation*}
It follows that 
\begin{equation}
\label{eq:Lag}
L(q, v)=\frac{1}{2}g_{ij}v^iv^j+\eta_i v^i-V(q).
\end{equation}
It is an exercise in the chain rule to show that the Euler-Lagrange equations in bundle coordinates \eqref{eq:EL-bundle} are equivalent
to 
 \begin{equation}
\label{euler-lagrange}
 \frac{dq^i}{dt}=\rho^i_j v^j, \qquad 
 \frac{d}{dt}\left(  \frac{\partial L}{\partial v^i}\right)
= \rho^j_i  \frac{\partial L}{\partial q^j} -
C^k_{ij} v^j  \frac{\partial L}{\partial v^k}, \qquad i=1,\dots,n,
\end{equation}
where $C^k_{ij}\in C^{\infty}(U)$ are the {\em  local structure functions}, determined in terms of 
the  Lie brackets of the frame vector fields by
\begin{equation*}
 [e_i, e_j](q)=C^k_{ij}(q) e_k(q).
\end{equation*}
The reader can also find Eqs (\ref{euler-lagrange}) in \cite{LeMaMa}.

%

\subsubsection*{The Euler-Lagrange equations as the trajectories of a second order vector field}

The regularity of the Lagrangian (i.e. the invertibility of the Hessian matrix $\frac{\partial^2 L}{\partial v^i \partial v^j}=g_{ij}$) implies
that  the Euler-Lagrange equations determine the integral curves of  a vector field  $\Gamma_L \in \mathfrak{X}(TQ)$. Using  (\ref{euler-lagrange}) it
is easily shown that locally 
\begin{equation}
\label{eq:EL-local1}
\Gamma_L=\rho^i_j v^j\frac{\partial}{\partial q^i}
+ g^{ij}\left(\rho^k_j \frac{\partial L}{\partial q^k} 
-\rho^k_lv^l\frac{\partial^2 L}{\partial q^k\partial v^j}
-
C^k_{jl} v^l\frac{\partial L}{\partial v^k}\right) \frac{\partial}{\partial v^i},
\end{equation}
where $g^{ij}$ are the entries of the matrix inverse of $g_{ij}$. 
%
%
The vector field   $\Gamma_L$ is {\em second order}. Namely, it satisfies:
\[
(T_{v_q}\tau_Q)(\Gamma_L(v_q)) = v_q, \mbox{ for } v_q \in T_qQ.
\]
As a consequence, the integral curves of $\Gamma_L$ are the tangent lifts of curves on $Q$. 

\subsubsection*{Effect of the gyroscopic 1-form}
Write  $\Gamma_L=\Gamma_{L_{mech}}-\Gamma_\eta,$ where $\Gamma_{L_{mech}}\in {\mathfrak X}(TQ)$ is the vector field corresponding to 
 $L_{mech}$
and $\Gamma_\eta$ is the {\em vector gyroscopic force}.
Using \eqref{eq:Lag} and \eqref{eq:EL-local1} one derives the local expressions:
\begin{eqnarray*}
\Gamma_{L_{mech}}
&=&\rho^i_j v^j\frac{\partial}{\partial q^i}
+ g^{ij} \left (\left(\frac{1}{2}\rho^k_j \frac{\partial g_{lm}}{\partial q^k} 
-\rho^k_l\frac{\partial g_{jm}}{\partial q^k}
- C^k_{jl} g_{km} \right ) v^lv^m-  \rho^k_j\frac{\partial V}{\partial q^k} \right)\frac{\partial}{\partial v^i},
\\
 \Gamma_\eta&=&g^{ij}v^k\left( \rho^l_k \frac{\partial \eta_i}{\partial q^l}-\rho^l_i\frac{\partial \eta_k}{\partial q^l}-\eta_lC^l_{ki}\right)\frac{\partial}{\partial v^j}.
\end{eqnarray*}
As is well-known,  $\Gamma_{L_{mech}}$ is intrinsically given by $\Gamma_{L_{mech}}=\Gamma_g- (\hbox{grad}_g V)^{\mathbf V}$,
where  $\Gamma_g\in {\mathfrak X}(TQ)$ is the geodesic flow of the metric $g$, and $ (\hbox{grad}_g V)^{\mathbf V}$ is the {\em 
potential vector force} given by the vertical
lift of the $g$-gradient of the 1-form $dV$, that is\footnote{We employ the standard {\em musical} notation for the  bundle isomorphisms induced by the metric $g$. Namely 
$\flat_g: TQ\rightarrow T^*Q$ is defined by 
$\langle \flat_g (v_q), w_q\rangle=g_q(v_q,w_q)$, for all $v_q, w_q\in T_qQ$ and its inverse is  $\sharp_g=(\flat_g)^{-1}: T^*Q\rightarrow TQ$.}
\begin{equation*}
\hbox{grad}_g V=\sharp_g(dV), \qquad  (\hbox{grad}_g V)^{\mathbf V}(u_q)=\frac{d}{dt}\Big|_{t=0} 
   \left(u_q + t (\hbox{grad}_g V)(q)\right),\quad \forall u_q\in T_qQ.
\end{equation*}
On the other hand, an intrinsic interpretation of the 
  gyroscopic vector force $\Gamma_\eta$ may be given by considering the  $(1,1)$-tensor field $\Teta$ 
  on $Q$ defined in terms of $d\eta$ and the metric $g$ by
   $$
   \Teta(u_q)=\sharp_g(i_{u_q}d\eta(q))\;,\quad u_q\in T_qQ. 
   $$
 We have $\Gamma_\eta=(\Teta)^{\mathbf V}$ where, as usual,
  $$
  (\Teta)^{\mathbf V}(u_q)=\frac{d}{dt}\Big|_{t=0} 
   \left(u_q+t \Teta(u_q)\right),\qquad \forall u_q\in T_qQ\; .
   $$
This description of $\Gamma_\eta$  shows that a closed gyroscopic 1-form has no effect  
on the equations of motion. In such a  case  $L$ and $L_{mech}$ determine the same 
equations of motion and are called equivalent Lagrangians (see \cite{Krupkova}).

\subsubsection*{Energy conservation}

 The {\em energy function} $E_L\in C^\infty (TQ)$ is intrinsically defined as
$
E_L=\Delta L-L
$,
where $\Delta\in {\mathfrak X}(TQ)$ is the Liouville vector field given by
\[
\Delta(u)=\frac{d}{dt}\Big|_{t=0}(u+tu)\in T_uTQ,
\]
or, in coordinates,
$\Delta=v^i\frac{\partial}{\partial v^i}$. 
Interestingly,  the gyroscopic 1-form $\eta$ does not contribute to the energy, which is the sum of the kinetic and potential energies.  Indeed, as one easily computes
$E_L=E_{L_{mech}}={\mathcal K}_g+V\circ\tau_Q$, or, in 
 coordinates, 
\begin{equation}\label{Lagrangian-energy-local}
E_L(q, v)=\frac{1}{2}g_{ij}(q)v^iv^j+V(q)\; .
\end{equation}
As is well known, the time independence of  $L$ implies that $E_L$ is a first integral of $\Gamma_L$.

\subsection{The nonholonomic case}\label{sec:lagra-constrained}

As mentioned above, quasivelocities are very convenient in the presence of nonholonomic constraints.
We  assume that the constraints are linear in the velocities, that is, they are given by a smooth  distribution ${\mathcal D}\hookrightarrow TQ$ which we assume to
have constant rank $r\geq 2$.  As is well known the constraints are nonholonomic if and only if $\D$ is 
non-integrable. A fundamental role in our treatment will be played by the inclusion and projection maps
\begin{equation}
\label{eq:iP}
i_\D:\D\hookrightarrow TQ, \qquad \Pp:TQ\to \D.
\end{equation}
Here $\Pp$ is the  bundle projector induced from the orthogonal bundle decomposition $TQ=\D\oplus \D^\perp$, where
$\D^\perp:=\{u\in TQ\; |\, g(u, v)=0, \ \hbox{for all} \ v\in {\mathcal D}\}$.

\subsubsection*{Local Expressions}
We work with a  local frame  adapted to the constraints. More precisely, our frame is 
\begin{equation}\label{adaptedbasis}
\{e_i\} = \{e_a, e_{\alpha}\}, \quad 1\leq a\leq r, \quad r+1\leq \alpha\leq n,
\end{equation}
where $\{e_a\}$ is a local basis of sections of $\D$ and $\{e_{\alpha}\}$ 
is a local basis of sections of $\D^\perp$. 
The  induced coordinates $(q^i, v^a, v^{\alpha})$ of $TQ$ are such  that 
\begin{equation}
\label{eq:Dcoords}
{\mathcal D}=\{(q^i, v^a, v^{\alpha})\; |\, v^{\alpha}=0, \  r+1\leq \alpha\leq n\}.
\end{equation}
In other words,  the nonholonomic constraints in these coordinates  become  $v^{\alpha}=0$, $r+1\leq \alpha\leq n$, and we may thus
use $(q^i,v^a)$ as coordinates for the velocity phase space ${\mathcal D}$.

D'Alembert's principle  leads to the nonholonomic equations of motion  (see e.g.
\cite{BLMMM2011,BMZ2005,CdLMM}): 
\begin{equation*}
\begin{split}
\displaystyle \frac{dq^i}{dt}&=\rho^i_a v^a, \qquad i=1,\dots, n, \\
\displaystyle\frac{d}{dt}\left( \displaystyle \frac{\partial L}{\partial v^a}\right)
&=\rho^j_a \displaystyle \frac{\partial L}{\partial q^j} -
C^c_{ab} v^b\displaystyle \frac{\partial L}{\partial v^c}-
C^\alpha_{ab} v^b\displaystyle \frac{\partial L}{\partial v^\alpha} , \qquad a=1,\dots, r, \\
v^{\gamma}&=0, \qquad \gamma=r+1,\dots, n,
\end{split}
\end{equation*}
where $L$ is given by  \eqref{eq:Lag}, and, in accordance with our index convention, 
\begin{equation}
\label{eq:Cabc}
[e_a, e_b](q)=C_{ab}^c(q) e_c+C_{ab}^{\gamma}(q) e_{\gamma}(q).
\end{equation}

It is convenient to write the  above  equations of motion  in terms of the {\em constrained Lagrangian},  $L_c:=L\circ i_\D :\D\to \R$.  In view of \eqref{eq:Lag} and \eqref{eq:Dcoords}, it has local expression
\begin{equation}
\label{eq:LagC}
L_c(q, v^a)=\frac{1}{2}g_{ab}v^av^b+\eta_a v^a-V(q).
\end{equation}
A simple calculation shows that
\begin{equation}
\label{nonholonomic-equations1}
\begin{split}
 \frac{dq^i}{dt}&=\rho^i_a v^a, \qquad i=1,\dots, n,\\
\frac{d}{dt}\left(  \frac{\partial L_c}{\partial v^a}\right)
&=\rho^j_a \displaystyle \frac{\partial L_c}{\partial q^j} -
C^c_{ab} v^b \frac{\partial L_c}{\partial v^c}-C_{ab}^\alpha v^b \eta_\alpha,\qquad a=1,\dots, r.
\end{split}
\end{equation}
 
\subsubsection*{Intrinsic geometry of the equations}

The solutions of \eqref{nonholonomic-equations1}  are the integral curves of the nonholonomic vector field $\Gamma_{(L, {\mathcal D})}\in {\mathfrak X}({\mathcal D})$, which is   intrinsically characterised by the conditions
\begin{equation}
\label{relation-nonolonomic-free}
 (T_{v_q}i_\D)   (\Gamma_{(L, {\mathcal D})}(v_q))- \Gamma_L (v_q) \in \F_{v_q},\qquad
    \Gamma_{(L, {\mathcal D})}(v_q) \in T_{v_q}\D, \qquad \forall v_q\in \D.
\end{equation} 
Here $T_{v_q}i_\D: T_{v_q}\D \to T_{v_q}(TQ)$ is the tangent space inclusion  and the subspace $ \F_{v_q}\subseteq T_{v_q}(TQ)$, which contains the constraint reaction force, is determined according to D'Alembert's principle by  
$$ \F_{v_q}=\left(\left[ \sharp_g(\D^\circ)\right]^{\mathbf V}\right)(v_q)
=\left\{ \frac{d}{dt}\Big|_{t=0} (v_q+t\, \sharp_g(\gamma_q))\; |\; \gamma_q\in {\mathcal D}_q^\circ\right\},
$$
where 
$\D_q^\circ\subseteq T_q^*Q$ is the annihilator of $\D_q$.
The nonholonomic vector field $\Gamma_{(L, {\mathcal D})}$ so determined is second order with 
respect to the bundle structure $(\tau_Q)_{|{\mathcal D}}: {\mathcal D} \to Q$, namely,
\begin{equation*}
T_{v_q} \tau_Q(\Gamma_{(L, {\mathcal D})}(v_q))=v_q, \qquad \forall v_q\in {\mathcal D}_q.
\end{equation*}

\subsubsection*{Effect of the gyroscopic 1-form}

As in the unconstrained case,  the effect of the gyroscopic 1-form $\eta$ may be identified by   decomposing the nonholonomic vector field $\Gamma_{(L, {\mathcal D})}=\Gamma_{(L_{mech}, {\mathcal D})}-\Gamma_{(\eta, \D)}$. 
In view of \eqref{nonholonomic-equations1} and \eqref{eq:LagC} we have the local expressions
\begin{equation}\label{Local-constrained-SOVF}
\begin{split}
\Gamma_{(L_{mech}, {\mathcal D})}
&=\rho^i_a v^a\displaystyle \frac{\partial}{\partial q^i}
+ g^{cd}\left[v^av^b\left(\displaystyle \frac{1}{2}\rho^k_d \displaystyle \frac{\partial g_{ab}}{\partial q^k} 
-\displaystyle \rho^k_a\frac{\partial g_{db}}{\partial q^k}
- C^e_{ad} g_{eb}\right)-\rho^k_d\displaystyle \frac{\partial V}{\partial q^k}\right]\displaystyle \frac{\partial}{\partial v^c}, \\ 
 \Gamma_{(\eta, \D)} &= g^{cd}v^a\left( \rho^i_a \displaystyle \frac{\partial \eta_d}{\partial q^i}-\rho^i_d\displaystyle \frac{\partial \eta_a}{\partial q^i}-\eta_iC^i_{ad}\right)\displaystyle \frac{\partial}{\partial v^c}\; .
\end{split}
\end{equation}
An intrinsic expression for the constrained  gyroscopic vector force $ \Gamma_{(\eta, \D)}$ can be given
in terms of the  vector bundle endomorphism $\mathcal{P} \circ \Teta\circ i_\D: \D \to \D$. We have
\begin{equation*}
 \Gamma_{(\eta, \D)}(u_q) =(\mathcal{P} \circ \Teta\circ i_\D)^{\mathbf V}(u_q) : = \frac{d}{dt}\Big|_{t=0} (u_q+t (\mathcal{P} \circ \Teta \circ i_\D)(u_q)).
\end{equation*}
The above equality may be verified using \eqref{Local-constrained-SOVF} and noting that, due  to our choice of frame, the orthogonal projector $\mathcal{P}: TQ \to \D$ satisfies $\mathcal{P}(q^i, v^a, v^{\alpha})=(q^i, v^a)$.

\subsubsection*{Energy conservation}
Since the   constraints are linear and time independent, the energy is also preserved along the constrained motion. More precisely, the restricted energy function  $E_c:=E_L\circ i_\D \in C^\infty(\D)$ is a first integral of  the nonholonomic vector field $\Gamma_{(L, {\mathcal D})}$. In local coordinates we have
\begin{equation}\label{constrained-energy}
E_c(q, v^a)=\frac{1}{2}g_{ab}(q)v^av^b+V(q).
\end{equation}
\section{The Legendre transformation and the Hamiltonian formalism}
\label{sec:Ham}

\subsection{The free case}
The Legendre transformation 
${\mathbb F}L: TQ\rightarrow T^*Q$
associated with $L: TQ\rightarrow {\mathbb R}$
is defined as the fiber derivative
\begin{equation}\label{eq:Leg-transf}
\langle {\mathbb F}L(u_q), v_q\rangle :=\frac{d}{dt}\Big|_{t=0} L(u_q+tv_q),
\end{equation}
for all $u_q, v_q\in T_qQ$. 
For a Lagrangian of the type (\ref{def:lagrangian}) we have 
\begin{equation}\label{Expression-Legendre-transformation}
{\mathbb F}L=\flat_g+\eta\circ \tau_Q.
\end{equation}
A fundamental observation is that 
 the presence of the gyroscopic 1-form $\eta$, makes ${\mathbb F}L$ into an {\em affine} bundle isomorphism 
(instead of a linear one).

As usual, the  Hamiltonian function $H\in C^\infty(T^*Q)$ associated with $L$ is defined  by
$H:=E_L\circ ({\mathbb F}L)^{-1}$.
The dynamics in this formulation is given by the  Hamiltonian vector field $X_H\in  \mathfrak{X}(T^*Q)$ 
characterised as the unique vector field satisfying 
\[
X_H(F)= \{ F,H\}, \qquad \mbox{for all} \quad F\in C^\infty(T^*Q),
\]
where $\{ \cdot ,\cdot \},$ is the canonical Poisson bracket on $T^*Q$ \cite{AM78,LR89}. The well-known 
equivalence of the Lagrangian and Hamiltonian 
formulations of the dynamics is expressed as 
\begin{equation}\label{equivalence-dynamics}
({\mathbb F}L)_*\Gamma_L=X_H.
\end{equation}
The Hamiltonian vector field $X_H$ is second order in the sense that ${\mathbb F}L\left( (T_{\alpha_q}\pi_Q)(X_H(\alpha_q))\right)=\alpha_q$ for
all $\alpha_q\in T^*Q$, where $\pi_Q:T^*Q\to Q$ is the cotangent bundle projection.

\subsubsection*{Local expressions}
Any covector $\gamma_q\in T_q^*Q$ may be expressed uniquely 
in terms of our co-frame  $\{e^i\}$ as a linear combination $\gamma_q=p_ie^i(q)$.
The coefficients $p_i$ are the cotangent bundle analog of the quasi-velocities $v^i$, and  will be called {\em quasi-momenta} (although
this terminology is not very standard). If $(q^i)$ are local coordinates for $Q$, then  $(q^{i},p_i)$ are local coordinates
for  $T^*Q$ which in general 
 are {\em not} Darboux coordinates.

The local expressions for the Legendre transformation ${\mathbb F}L$ and its inverse are:
\begin{equation}
\label{eq:Legendre-local-full}
p_i=g_{ij}v^j+\eta_i, \qquad v^i=g^{ij}(p_j-\eta_j).
\end{equation}
Therefore, using \eqref{Lagrangian-energy-local}, we obtain  the local expression for $H$:
\begin{equation}
\label{eq:Ham-specific}
H(q, p)=\displaystyle \frac{1}{2}g^{ij}(p_i - \eta_i)(p_j-\eta_j) +V(q).
\end{equation}
Note that, in contrast with the energy function $E_L\in C^\infty(TQ)$ which is independent of the gyroscopic 1-form $\eta$, the affine 
nature of the Legendre transformation leads to the appearance of $\eta$ in the Hamiltonian $H$.

The local expression for the Hamiltonian vector field in these coordinates is
\begin{equation}
\label{eq:HamVF-general}
X_H= \displaystyle \rho^j_i\frac{\partial H}{\partial p_i}\displaystyle \frac{\partial}{\partial q^j}-
\left( \rho^i_j\displaystyle \frac{\partial H}{\partial q^i}+C^k_{ji} p_k \displaystyle \frac{\partial H}{\partial p_i}\right) \displaystyle \frac{\partial}{\partial p_j}.
\end{equation}
The equivalence of the Lagrangian and Hamiltonian formulations can be directly verified using 
 \eqref{eq:Legendre-local-full}. The second order nature of $X_H$ corresponds to the identities $\frac{dq^j}{dt}=\rho^j_iv^i$ which are valid for all $j$,
and which hold in virtue of \eqref{eq:HamVF-general}, \eqref{eq:Ham-specific} and \eqref{eq:Legendre-local-full} (see also \cite{LeMaMa}).

\subsection{The nonholonomic case}
\label{ss:nonho-Ham}


The passage from the Lagrangian to the Hamiltonian formulation in the constrained case
 is  given by the {\em constrained Legendre transformation}
$\FLc:\D \to \D^*$ which we define as  the fiber derivative of the constrained Lagrangian $L_c:\D \to \R$. Namely,
\[
\langle \FLc (u_q), v_q\rangle :=\frac{d}{dt}\Big|_{t=0} L_c(u_q+tv_q),
\]
for all $u_q, v_q\in \D_q$. The constrained Legendre transformation is a diffeomorphism. More precisely,
and as in the free case, the gyroscopic 1-form  makes $\FLc$ into an invertible   bundle map which is generally affine instead of linear (see Proposition \ref{prop:ConstLegTransf} below).

The {\em constrained Hamiltonian} function $H_c\in C^\infty( {\mathcal D}^*)$ and 
the {\em nonholonomic vector field} $\XHc\in \frak{X}(\D^*)$ describing the dynamics on $\D^*$ are defined by pushing forward with $\FLc$
the corresponding objects on $\D$: 
\begin{equation}\label{Def-H-c}
H_c := E_c \circ (\FLc)^{-1}, \qquad \XHc :={\FLc}_*(\Gamma_{(L, {\mathcal D})}).
\end{equation}
In analogy with the unconstrained case, the nonholonomic vector field $\XHc$ is second order meaning 
 that ${\mathbb F}L_c \left( T_{\alpha_q} \tau_{\D^*}(\XHc(\alpha_q))\right)=\alpha_q$ for
all $\alpha_q\in \D^*$, where $ \tau_{\D^*}:\D^*\to Q$ is the vector bundle projection.

\subsubsection*{A distribution on $\D^*$}

Following previous references e.g. \cite{Weber, BatesSniaticky, CdLMM,Cu-Du}, we consider the   distribution $\E$ on $\D^*$  defined as
\begin{equation}
\label{eq:distE}
\E_{\alpha_q}:=\{ v\in T_{\alpha_q}\D^* \, : \, T_{\alpha_q} \tau_{\D^*}(v) \in \D_q\},
\end{equation}
for $\alpha_q\in \D^*$ and where, as usual, $ \tau_{\D^*}:\D^*\to Q$ is the vector bundle projection. It is known that, under our hypothesis,
$\E$ is a regular  distribution of rank $2r = 2\,\mathrm{rank} (\D)$  
which is integrable if and only if $\D$ is integrable and has the property 
that $\XHc$ takes values on $\E$, i.e. $\XHc(\alpha_q)\in \E_{\alpha_q}$ for all $\alpha_q\in \D^*$. Moreover,
 $\E$ is symplectic meaning that $\E_{\alpha_q}$ is a symplectic subspace of $T_{{\mathcal P}^*(\alpha_q)}(T^*Q)$ for all $\alpha_q\in \D^*$, where ${\mathcal P}^*: \D^*\hookrightarrow T^*Q$ is the dual morphism of the orthogonal projector ${\mathcal P}: TQ\rightarrow \D$.

%
\subsubsection*{The metric decomposition  of $T^*Q$}

The orthogonal decomposition  $TQ=\D\oplus \D^\perp$, associated to the kinetic energy, induces
a bundle decomposition of the cotangent bundle in terms of the annihilators: $T^*Q= (\D^\perp)^\circ \oplus \D^\circ$. Moreover,
the kinetic energy induces  
  natural identifications 
  \begin{equation}
  \label{eq:D^*id}
   \D^*\simeq  (\D^\perp)^\circ =\flat_g(\D)\qquad  \mbox{and} \qquad (\D^\perp)^*\simeq  \D^\circ =\flat_g(\D^\perp),
\end{equation}
 which allow us to write 
 \begin{equation}
 \label{eq:T*Qdecomp}
T^*Q= \D^* \oplus  (\D^\perp)^*.
\end{equation}
In our treatment below we will make extensive use of the bundle maps  
\begin{equation}
\label{eq:co-projinc}
i_\D^*:  T^*Q \to \D^*, \qquad \mathcal{P}^*:  \D^* \hookrightarrow T^*Q, 
\end{equation}
which are dual to the maps $i_\D$ and $\Pp$ defined in Eq \eqref{eq:iP} above.
In view of our identifications, it is easy to check that $\mathcal{P}^*$ is the inclusion of  $\D^*$ on $T^*Q$ and instead
 $ i_\D^*$  the projection onto $\D^*$ according to the  decomposition \eqref{eq:T*Qdecomp}.

 According to the decomposition  \eqref{eq:T*Qdecomp} we may uniquely write the
 gyroscopic 1-form as 
 \begin{equation*}
\eta = \eta^\parallel+ \eta^\perp, \quad \mbox{with}\quad   \eta^\parallel(q)\in \D^*_q, \;\;\; \eta^\perp(q)\in  (\D^\perp)_q^*, \;\;\; \forall q\in Q.
\end{equation*}

\begin{prop} 
\label{prop:ConstLegTransf}
The constrained Legendre transformation $\FLc:\D \to \D^*$ is explicitly given by
\begin{equation}
\label{eq:ConstLegTransf} 
 \FLc =i_{\D}^*\circ( \flat_g  +  \eta^\parallel \circ \tau_Q) \circ i_\D.
\end{equation}
In particular $\FLc$ is an affine  bundle isomorphism which is linear if and only if $ \eta^\parallel=0$.
\end{prop}
\begin{proof}
Let $u_q, v_q\in \D_q$. Given that $u_q+t v_q\in \D_q$ we have $L_c(u_q+t v_q)=L(u_q+t v_q)$ and therefore,
using \eqref{eq:Leg-transf} and \eqref{Expression-Legendre-transformation}, we have
\[
\langle \FLc (u_q), v_q\rangle =\frac{d}{dt}\Big|_{t=0} L(u_q+tv_q)=\langle {\mathbb F}L (u_q), v_q\rangle = \langle i_{\D}^*(\flat_g(u_q)+\eta(q)), v_q\rangle .
\]
But, $\langle i_{\D}^*(\eta(q)), v_q\rangle=\langle \eta^\parallel (q), v_q\rangle$ since $v_q\in \D_q$. The conclusion about the nature of  $\FLc$ 
is obvious from  \eqref{eq:ConstLegTransf}.
\end{proof}

\subsubsection*{The nonholonomic dynamics as a projection of the free Hamiltonian vector field}

The following theorem is fundamental for our purposes of expressing the vector field $\XHc\in \frak{X}(\D^*)$ in almost-Hamiltonian form.
It states that $\XHc$ equals an appropriate projection of unconstrained Hamiltonian vector fields. In its
statement, $ \mathbb{F}L(\D)$ denotes the image of $\D$ under the (unconstrained) Legendre transformation. In view 
of \eqref{Expression-Legendre-transformation},
$ \mathbb{F}L(\D)$ equals the affine subbundle $\D^*+\eta$ of $T^*Q$.

\begin{thm}\label{teo1}
Let $\beta_q \in \mathbb{F}L(\D_q)$. The following equalities hold 
\begin{equation}
\label{eq:thm1}
\begin{split}
\XHc(i_\D^*(\beta_q)) = (T_{\beta_q}i_\D^*)(X_H(\beta_q))  = (T_{\beta_q}i_\D^*)(X_{H_c\circ  i_\D^*} (\beta_q)).
\end{split}
\end{equation}
\end{thm}
\begin{proof}
In the next section we will give the local expressions for $H$ and $H_c\circ  i_\D^*$ from which it is easily verified 
that $dH(\beta_q)=d(H_c\circ  i_\D^*)(\beta_q)$ (Lemma \ref{l:equaldiff}). As a consequence $X_H(\beta_q)=X_{H_c\circ  i_\D^*} (\beta_q)$
and  the second equality of \eqref{eq:thm1}  holds. To prove the first equality, let  $u_q \in  \D_q$ such that $\beta_q= {\mathbb F}L(u_q)$. From the intrinsic version of the Lagrange-D'Alembert Principle 
(\ref{relation-nonolonomic-free}), we have
\[
(T_{u_q}i_\D)\Gamma_{(L, {\mathcal D})}(u_q)-\Gamma_L(u_q) = \frac{d}{dt}\Big|_{t=0}(u_q+t\sharp_g(\gamma_q)), \quad \mbox{for some} \quad \gamma_q\in {\mathcal D}_q^\circ.
\]
 Applying $T_{u_q}{\mathbb F}L$ to both sides of the equation 
and using expression (\ref{Expression-Legendre-transformation}) for ${\mathbb F}L$ yields,
\begin{eqnarray*}
(T_{u_q}({\mathbb F}L \circ i_\D))(\Gamma_{(L, {\mathcal D})}(u_q))-(T_{u_q}{\mathbb F}L)(\Gamma_L(u_q))
&=& \frac{d}{dt}\Big|_{t=0}({\mathbb F}L (u_q+t\sharp_g(\gamma_q)))\\
&=&\frac{d}{dt}\Big|_{t=0}(\flat_g(u_q)+t\gamma_q+\eta(q))\\
&=&\frac{d}{dt}\Big|_{t=0}(\beta_q+t\gamma_q).\\
\end{eqnarray*}
Now notice that $(T_{u_q}{\mathbb F}L)(\Gamma_L(u_q))=X_H(\beta_q)$  by (\ref{equivalence-dynamics}). Thus, applying $T_{\beta_q}i_\D^*$ on both sides of the above equation gives
\[
T_{u_q}(i^*_{\mathcal D}\circ {\mathbb F}L \circ i_{\D})(\Gamma_{(L, {\mathcal D})}(u_q))-(T_{\beta_q}i^*_{\mathcal D})(X_H(\beta_q))=\frac{d}{dt}\Big|_{t=0}(i_{\mathcal D}^*\beta_q+t i_{\mathcal D}^*\gamma_q),
\]
but the right hand side vanishes since  $i_{\mathcal D}^*\gamma_q=0$ because $\gamma_q\in {\mathcal D}_q^\circ=(\D^\perp)_q^*$. 
Therefore, we conclude that
\begin{equation}
\label{eq:aux-proof-thm1}
T_{u_q}(i^*_{\mathcal D}\circ {\mathbb F}L \circ i_{\D})(\Gamma_{(L, {\mathcal D})}(u_q))= (T_{\beta_q}i^*_{\mathcal D})(X_H(\beta_q)).
\end{equation}
Now, using \eqref{Expression-Legendre-transformation} and \eqref{eq:ConstLegTransf} we conclude  that $i^*_{\mathcal D}\circ {\mathbb F}L \circ i_{\D}=\FLc$. This has two  consequences, the first is that  
$\FLc(u_q)= i^*_{\mathcal D}( {\mathbb F}L (u_q))= i^*_{\mathcal D}(\beta_q)$, and the second is that 
 the left hand side of \eqref{eq:aux-proof-thm1} may
be rewritten as 
\begin{equation*}
\begin{split}
T_{u_q}(i^*_{\mathcal D}\circ {\mathbb F}L \circ i_{\D})(\Gamma_{(L, {\mathcal D})}(u_q))&=(T_{u_q}\FLc)(\Gamma_{(L, {\mathcal D})}(u_q)) \\
&=\XHc(\FLc(u_q)) \\
&=\XHc(i^*_{\mathcal D}\beta_q),
\end{split}
\end{equation*}
where we have used the definition \eqref{Def-H-c} of $\XHc$ in the second equality.
Hence, Eq. \eqref{eq:aux-proof-thm1} implies that  $\XHc(i^*_{\mathcal D}\beta_q)=(T_{\beta_q}i^*_{\mathcal D})(X_H(\beta_q))$ as required.
\end{proof}

\subsubsection*{Local expressions}
Our choice of frame in \eqref{adaptedbasis} implies that the dual co-frame $\{e^i\}$ also splits as 
 $\{e^{i}\} = \{e^{a}, e^{\gamma}\}$ and is adapted to the decomposition \eqref{eq:T*Qdecomp}. We write the corresponding quasi-momenta as $( p_{i}) = (p_{a}, p_{\gamma})$.
 It follows that $(q^i, p_{a}, p_{\gamma})$ are coordinates for $T^*Q$ and 
 $(q^i,p_a)$ are coordinates for $\D^*$ which is described as a submanifold of $T^*Q$ 
 by the conditions $p_\gamma=0$, $r+1\leq \gamma \leq n$. In these coordinates, and with the conventions introduced in the previous 
 Section \ref{sec:lagra-constrained}, we have 
%
%
\[
i_\D(q^{i}, v^{a}) = (q^{i}, v^{a}, 0), \qquad i_\D^*(q^{i}, p_{a}, p_\gamma) = (q^{i}, p_a).
\]
The components $\eta^\parallel$ and $\eta^\perp$ of the gyroscopic form $\eta$ are
\begin{equation*}
\eta^\parallel = \eta_a e^a, \qquad \eta^\perp = \eta_\gamma e^\gamma.
\end{equation*}
Hence, in view of \eqref{eq:ConstLegTransf}, the constrained Legendre transformation  refines identities \eqref{eq:Legendre-local-full}
to
\begin{equation}
\label{eq:Legendre-local-constrained}
p_a=g_{ab}v^b+\eta_a, \qquad v^b=g^{ab}(p_a-\eta_a).
\end{equation}
%
Combining the second equality above  with the local expression \eqref{constrained-energy} for $E_c$, yields the following local
expression for the constrained Hamiltonian $H_c$
\begin{equation}\label{Local-Expression-H-c}
\begin{array}{rcl}
H_c(q^i, p_a)= \displaystyle \frac{1}{2}g^{ab}(p_a-\eta_a)(p_b-\eta_b) + V(q).
\end{array}
\end{equation}

Now, it is easy to verify that the following standard relations between the constrained Lagrangian $L_c$ (given by \eqref{eq:LagC})
and the constrained Hamiltonian $H_c$ given above hold
\begin{equation*}
p_a=\frac{\partial L_c}{\partial v^a}, \qquad v^a=\frac{\partial H_c}{\partial p_a},\qquad 
\frac{\partial H_c}{\partial q^i}=-\frac{\partial L_c}{\partial q^i}.
\end{equation*}
 Hence, under the constrained Legendre transformation \eqref{eq:Legendre-local-constrained}, the equations of motion \eqref{nonholonomic-equations1}
 transform to 
 \begin{equation}
\label{nonholonomic-equations2}
\begin{split}
 \frac{dq^i}{dt}&=\rho^i_a \frac{\partial H_c}{\partial p_a}, \qquad i=1,\dots, n,\\
\frac{d p_a}{dt}
&=-\rho^j_a \displaystyle \frac{\partial H_c}{\partial q^j} -
C^d_{ab} p_d \frac{\partial H_c}{\partial p_b}-
C^\alpha_{ab} \eta_\alpha \frac{\partial H_c}{\partial p_b},\qquad a=1,\dots, r.
\end{split}
\end{equation}
In particular  we have $\frac{dq^i}{dt}=\rho^i_ag^{ab}(p_b-\eta_b)=\rho^i_av^a$ for all $i$, which 
expresses the second order nature of $\XHc$ whose   local expression is
\begin{equation}\label{Local-expression-terms-H-c}
\XHc = \left(\rho^j_a \displaystyle \frac{\partial H_c}{\partial p_a}  \right)\displaystyle\frac{\partial}{\partial q^j} 
- \left( \rho^i_a  \displaystyle \frac{\partial H_c}{\partial q^i}
+ (C^d_{ab}p_d+C^\gamma_{ab} \eta_{\gamma}) \displaystyle \frac{\partial H_c}{\partial p_b}\right)
 \displaystyle \frac{\partial}{\partial p_a}.
\end{equation}

The distribution $\E$ defined by \eqref{eq:distE} is specified in our coordinates as
\begin{equation*}
\E_{(q^i,p_a)}=\mathrm{span} \left \{\rho^i_a g^{ab}(p_b-\eta_b)\frac{\partial}{\partial q^i} \, , \, \frac{\partial}{\partial p_b} \right \}.
\end{equation*}
It is clear from this description that its rank equals $2r$ and $\XHc$ takes values on it.

We finish this section by using the above local expressions for $i_\D^*$ and  $H_c$ to prove the following Lemma that was used in the
proof of Theorem \ref{teo1}.

\begin{lemma}
\label{l:equaldiff}
$dH(\beta_q)= d(H_c\circ i_\D^*)(\beta_q)$ for all $\beta_q\in \mathbb{F}L(\D)_q$.
\end{lemma}
\begin{proof}
In view of \eqref{eq:Legendre-local-full} we have the following coordinate description of $\mathbb{F}L(\D)$:
\begin{equation*}
\mathbb{F}L(\D)= \left \{(q^i,p_a,p_\gamma) \, : \,  p_\gamma=\eta_\gamma, \quad \gamma=r+1,\dots, n \right \}.
\end{equation*}

On the other hand, using the local expressions for $i_\D^*$ and  $H_c$ given above we have
\begin{equation*}
H_c \circ i_\D^* (q^i,p_a,p_\gamma)= \displaystyle \frac{1}{2}g^{ab}(p_a-\eta_a)(p_b-\eta_b) + V(q).
\end{equation*}
Comparing with \eqref{eq:Ham-specific} we conclude that $(H-H_c \circ i_\D^* )(q^i,p_a,p_\gamma)=  \frac{1}{2}g^{\beta \gamma}(p_\beta-\eta_\beta)(p_\gamma-\eta_\gamma$). This expression vanishes quadratically along $\mathbb{F}L(\D)$ so its differential along 
this set also vanishes.
\end{proof}

\section{Almost-Poisson formulation}\label{sec-almost}

We are now ready to define the almost-Poisson bracket $\{\cdot, \cdot\}_{\D^*}$ for the nonholonomic  dynamics  $\XHc\in \frak{X}(\D^*)$ considered above. This is a bracket of functions on $\D^*$ 
such that 
$$\XHc(\varphi) = \{\varphi, H_c\}_{\D^*}, \qquad \mbox{for all $ \varphi \in C^\infty(\D^*)$}.  $$ 


\subsection{Definitions and main properties}

As before, denote by $\{\cdot, \cdot\}$  the canonical Poisson bracket  on $T^*Q$.
For $\varphi, \psi \in C^\infty(\D^*)$ we define $\{\varphi,  \psi\}_{\D^*}\in  C^\infty(\D^*)$ by
\begin{equation}
\label{def:constrained-bracket}
\{\varphi,  \psi\}_{\D^*} := \{\varphi \circ i_{\D}^*,  \psi \circ i_{\D}^*\} \circ (\mathcal{P}^* + \eta^\perp \circ \tau_{\D^*}),
\end{equation}
where $ \tau_{\D^*}:\D^*\to Q$ is the vector bundle projection.
To make sense of the above formula recall that  $i_\D^*$ and $\Pp^*+\eta^\perp \circ \tau_{\D^*}$ are bundle morphisms such that
\begin{equation}
\label{eq:bundle-morphisms}
i_\D^*:T^*Q \to \D^*, \qquad \mathcal{P}^* +\eta^\perp \circ \tau_{\D^*}:\D^*\to T^*Q.
\end{equation}

%
%
%

%


Note that the gyroscopic term manifests itself in   \eqref{def:constrained-bracket}  through
 the term involving $\eta^\perp$. In particular,
if the gyroscopic 1-form $\eta$ vanishes along $\D^\perp$, we recover previous constructions \cite{GLMM, LMD2010,LeLaLoMa2023} 
of the nonholonomic bracket for
mechanical Lagrangians.

The properties of the bracket $\{\cdot, \cdot\}_{\D^*}$ are given in Theorem  \ref{nonholonomic-bracket} below. In particular, as is clear 
from its definition, it is  ${\mathbb R}$-bilinear, skew-symmetric and satisfies 
 Leibniz rule in each argument. As a consequence, for $\varphi\in C^\infty( \D^*)$, the mapping $ \{\cdot, \varphi\}_{\D^*} : C^\infty(\D^*)\to  C^\infty(\D^*) $
  is a derivation which determines a vector field $X_\varphi^{\D^*} \in \frak{X}(\D^*)$ that we call the {\em almost-Hamiltonian vector field associated 
 to $\varphi$}, and is determined by
 \begin{equation*}
X_\varphi^{\D^*} (\psi)= \{\psi, \varphi\}_{\D^*} \qquad \forall \psi \in C^\infty(\D^*).
\end{equation*}

%
%

\subsubsection*{The bracket of linear and basic functions}

The following proposition expresses the value of the bracket \eqref{def:constrained-bracket}  applied to linear and basic functions
on $\D^*$. The resulting expressions  determine the bracket completely 
due to the vector bundle structure of $\D^*$. In the 
statement of the proposition,  $\Gamma(\D)$ is the space of sections of $\D$, and for $Y\in \Gamma(\D)$ we denote by $Y^\ell\in C^\infty(\D^*)$ the induced linear function $\alpha_q\mapsto Y^\ell(\alpha_q):=\langle \alpha_q, Y(q)\rangle$.
%
%
\begin{prop}\label{properties-constrained-bracket}
Let  $Y, Z \in \Gamma(\D)$ and $f, g \in C^\infty(Q)$. We have,
\begin{equation*}
\begin{split}
&\{Y^\ell, Z^\ell\}_{\D^*} = -(\mathcal{P}[Y, Z])^\ell + d\eta^\perp(Y, Z) \circ \tau_{\D^*}, \\
 &\{Y^\ell, f \circ \tau_{\D^*}\}_{\D^*} = -Y(f) \circ \tau_{\D^*}, \\
 & \{f \circ \tau_{\D^*}, g \circ \tau_{\D^*}\}_{\D^*} = 0.
\end{split}
\end{equation*}
\end{prop}
\begin{proof}
We rely on the following well-known properties of the canonical Poisson bracket:
\begin{equation}\label{properties-can-Poisson-bracket}
\{W_1^\ell, W_2^\ell\} = -[W_1, W_2]^\ell, \qquad \{W_1^\ell, f \circ \pi_Q\} = -W_1(f) \circ \pi_Q, \qquad \{f \circ \pi_Q, g \circ \pi_Q\} = 0,
\end{equation}
valid for $W_1, W_2\in \frak{X}(Q)$ where $\pi_Q:T^*Q\to Q$ is the cotangent bundle projection.
These,   together with   \eqref{def:constrained-bracket}, yield
\begin{align*}
\{Y^\ell, Z^\ell\}_{\D^*} &= -(\mathcal{P}[Y, Z])^\ell + \eta^\perp[Y, Z] \circ \tau_{\D^*}, \\
\{Y^\ell, f \circ \tau_{\D^*}\}_{\D^*} &= -Y(f) \circ \tau_{\D^*}, \ \; \{f \circ \tau_{\D^*}, g \circ \tau_{\D^*}\}_{\D^*} = 0.
\end{align*}
Now, since 
$\eta^\perp(Y) = \eta^\perp(Z) = 0$, we have $\eta^\perp[Y, Z] = -d\eta^\perp(Y, Z)$,
and the result follows.
\end{proof}

\subsubsection*{Main properties }

The main properties of the  bracket on $\D^*$ defined above are collected in the following theorem.  
In its formulation, the {\em characteristic distribution of the bracket} is the distribution on $\D^*$ 
generated by all  almost-Hamiltonian vector fields.

\begin{thm}\label{nonholonomic-bracket}
The bracket $\{\cdot, \cdot\}_{\D^*}$ defined by \eqref{def:constrained-bracket} satisfies: 
\begin{enumerate}
    \item[(i)] it is ${\mathbb R}$-bilinear,  skew-symmetric and satisfies the Leibniz rule in each argument;
    \item[(ii)] it is a fiberwise affine bracket on ${\mathcal D}^*$ which is fiberwise linear  if and only if
    \[
    d\eta^{\perp} (Y, Z)=0\; , \quad \forall \,Y, Z \in {\Gamma}({\mathcal D});
    \]
      \item[(iii)] its characteristic distribution coincides with $\E$ (defined by \eqref{eq:distE});
    \item[(iv)] the Jacobi identity is satisfied if and only if the constraint distribution ${\mathcal D}$ is integrable
    (i.e. if and only if the constraints are holonomic).
\end{enumerate}
\end{thm}
\begin{proof}
 $(i)$ is a direct consequence of the definition of the bracket as was mentioned above.

$(ii)$ If $\varphi, \psi \in C^\infty(\D^*)$ are fiberwise affine functions then they can be written as  
\[
\varphi = Y^\ell + f \circ \tau_{\D^*} \; \; \mbox{ and } \; \; \psi = Z^\ell + g \circ \tau_{\D^*},
\]
for $Y, Z \in \Gamma(\D)$ and $f, g \in C^\infty(Q)$. 
Hence, by Proposition \ref{properties-constrained-bracket}, we deduce that
\[
\{\varphi, \psi\}_{\D^*} = -(\mathcal{P}[Y, Z])^\ell + (d\eta^\perp(Y, Z) + Z(f) - Y(g)) \circ \tau_{\D^*},
\]
showing that $\{\varphi, \psi\}_{\D^*}$ is a fiberwise affine function on $\D^*$.

Moreover, for $Y, Z \in \Gamma(\D)$ we have
\[
\{Y^\ell, Z^\ell\}_{\D^*} = -(\mathcal{P}[Y, Z])^\ell + d\eta^\perp(Y, Z) \circ \tau_{\D^*},
\]
which is a fiberwise linear function on $\D^*$  if and only if $d\eta^\perp(Y,Z)=0$.
%

$(iii)$ The characteristic distribution  is generated by the almost-Hamiltonian vector fields of linear and basic functions
on $\D^*$.  Consider first the linear function $Y^\ell \in C^{\infty}({\mathcal D}^*)$ corresponding to 
a section $Y\in \Gamma({\mathcal D})$. \ Using Proposition \ref{properties-constrained-bracket},  we conclude that the
 its almost-Hamiltonian vector field  
 satisfies
\[
X^{\D^*}_{Y^\ell}(f \circ \tau_{{\mathcal D}^*})=\{ f\circ \tau_{{\mathcal D}^*}, Y^\ell\}_{{\mathcal D}^*}=Y(f)\circ \tau_{{\mathcal D}^*},
\]
for all $f\in C^{\infty}(Q)$
showing that  $X^{\D^*}_{Y^\ell} \in \frak{X}(\D^*)$ is $\tau_{{\mathcal D}^*}$-projectable over $Y\in \Gamma({\mathcal D})$. In particular this
 implies
that  $X^{\D^*}_{Y^\ell}$ takes values on $\E$.
Using again Proposition \ref{properties-constrained-bracket}, we find that the almost-Hamiltonian vector field $X^{\D^*}_{(g\circ \tau_{{\mathcal D}^*})}$
associated to the basic  function determined by  $g\in C^{\infty}(Q)$,
satisfies
\[
X^{\D^*}_{(g\circ \tau_{{\mathcal D}^*})}(f\circ \tau_{{\mathcal D}^*})=0, \; \; \forall f \in C^\infty(Q),
\]
showing that $X^{\D^*}_{(g\circ \tau_{{\mathcal D}^*})}$ is $\tau_{{\mathcal D}^*}$-vertical, and therefore it also takes values on $\E$.
This shows that the characteristic distribution is contained in $\E$. Using again Proposition \ref{properties-constrained-bracket} 
it is not difficult to argue that the characteristic distribution has rank 
$2\, \mathrm{rank}(\D)$ and therefore must coincide with $\E$. 

$(iv)$ A necessary condition for the validity of the Jacobi identity  is the integrability of the characteristic distribution. As was indicated right
after its definition,
$\E$ is integrable if and only if $\D$ is integrable, which implies that the bracket cannot satisfy the Jacobi identity if $\D$ is non-integrable.
%
%
Conversely, if we assume that ${\mathcal D}$ is integrable then we have $\mathcal{P}[Y,Z]=[Y,Z]$ for all $Y,Z\in \Gamma ({\mathcal D})$. This 
 implies that $d\eta^{\perp} (Y,Z)=0$, and  in particular, the first identity of 
Proposition \ref{properties-constrained-bracket} becomes $\{Y^\ell, Z^\ell\}_{\D^*} = -([Y, Z])^\ell$. The Jacobi identity can be easily proved using this
and the other formulas in Proposition \ref{properties-constrained-bracket}. 
\end{proof}

\subsection{Almost-Poisson description of the dynamics}
We now prove that the nonholonomic dynamics may be described in almost-Hamiltonian form with respect to the 
bracket $\{\cdot, \cdot\}_{\D^*}$ and the constrained Hamiltonian  $H_c\in C^\infty(\D^*)$. 
%
\begin{thm}
\label{th:HamForm}
The vector field $\XHc\in \mathfrak{X}(\D^*)$ describing the nonholonomic  dynamics satisfies
\[
\XHc(\varphi) = \{\varphi, H_c\}_{\D^*}, \; \; \mbox{ for all } \varphi \in C^\infty(\D^*).
\]
In other words, $\XHc$ coincides with the almost-Hamiltonian vector field $X^{\D^*}_{H_c}$ associated to $H_c\in C^\infty(\D^*)$.
\end{thm}
\begin{proof}
Let $\alpha_q \in \D_q^*$ and write $\beta_q:=(\mathcal{P}^* + \eta^\perp \circ \tau_{\D^*}) (\alpha_q)\in T_q^*Q$. On the one hand, using that $i_\D^*\circ \Pp^*$ is the identity
on $\D^*$ and $i_\D^*(\eta^\perp)=0$, we have $i_\D^*(\beta_q)=\alpha_q$. On the other hand, $\beta_q\in \mathbb{F}L(\D)_q$.
Indeed, from (\ref{Expression-Legendre-transformation}), we deduce that  $\beta_q =\mathbb{F}L( \sharp_g(\Pp^*(\alpha_q) - \eta^\parallel(q)))$ and 
since $\Pp^*(\alpha_q)-\eta^\parallel(q)\in (D^\perp_q)^\circ$ and $\flat_g(\D_q)=(\D_q^\perp)^\circ$,
one can check that 
$\sharp_g(\Pp^*(\alpha_q) - \eta^\parallel(q))\in \D_q$.

Therefore,  using Theorem \ref{teo1}  we have
\begin{equation*}
\begin{split}
\XHc(\varphi)(\alpha_q) &= \langle d\varphi (\alpha_q)  \, , \,  \XHc(\alpha_q)\rangle \\
&= \langle d\varphi (\alpha_q)  \, , \, (T_{\beta_q} i_{\D}^*)( X_{H_c \circ i_\D^*} (\beta_q))\rangle \\
&=  \langle T_{\beta_q}^* i_{\D}^*(d\varphi (\alpha_q)) \, , \, X_{H_c \circ i_\D^*} (\beta_q)\rangle \\
&= \langle d(\varphi \circ i_{\D}^*) (\beta_q) \, , \,  X_{H_c \circ i_\D^*} (\beta_q)\rangle \\
&= X_{H_c \circ i_\D^*}(\varphi \circ i_{\D}^*)(\beta_q).
\end{split}
\end{equation*}
Now use the fact that, by its definition, the Hamiltonian vector field $X_{F}\in \mathfrak{X}(T^*Q)$ of a function $F\in C^\infty(T^*Q)$ satisfies $X_{F}(G)=\{G,F\}$ for all $G\in C^\infty(T^*Q)$, to get
\begin{equation*}
\begin{split}
\XHc(\varphi)(\alpha_q) &= \{ \varphi \circ i_{\D}^*  ,  H_c \circ i_\D^* \}(\beta_q) \\
&=  \{ \varphi \circ i_{\D}^*  ,  H_c \circ i_\D^* \} \circ (\mathcal{P}^* + \eta^\perp \circ \tau_{\D^*}) (\alpha_q) \\
&=  \{ \varphi   ,   H_c  \}_{\D^*}(\alpha_q).
\end{split}
\end{equation*}
\end{proof}
\subsection{Local expressions}
As is well known, the canonical Poisson bracket $\{\cdot, \cdot\}$ on $T^*Q$ defines a  (canonical) Poisson $2$-vector $\Pi$ on $T^*Q$ by the relation $\{F , G\} = \Pi (dF, dG)$, for  $F, G \in C^\infty(T^*Q)$.
As may be proved, using for instance  (\ref{properties-can-Poisson-bracket}) (see also \cite{LeMaMa}), the local expression for $\Pi$ in our coordinates is
\[
\Pi = \rho^i_j\frac{\partial}{\partial q^i}\wedge \frac{\partial}{\partial p_j}
-\frac{1}{2}C^k_{ij}p_k\frac{\partial}{\partial p_i}\wedge \frac{\partial}{\partial p_j},
\]
or, in matrix form,
\[
\Pi =
\left(
\begin{array}{cc}
0&\rho^i_j\\[6pt]
-\rho^i_j&-C^k_{ij}p_k
\end{array}
\right).
\]
The Hamiltonian vector field $X_H$ of $H$, given previously by \eqref{eq:HamVF-general}, satisfies $X_H = \Pi(\cdot, d H)$, which corresponds to the following vector form of Hamilton's equations:
\[
\left(
\begin{array}{r}
\dot{q}^i\\
\dot{p}_j
\end{array}
\right)=
\left(
\begin{array}{cc}
0&\rho^j_i\\[6pt]
-\rho^j_i&-C^k_{ij}p_k
\end{array}
\right)
\begin{pmatrix}
\dfrac{\partial H}{\partial q^i}\\[8pt]
\dfrac{\partial H}{\partial p_j}
\end{pmatrix}.
\]

On  the other hand,   using Proposition \ref{properties-constrained-bracket} we find the following 
expressions for the nonholonomic brackets of our 
adapted coordinates $(q^i,p_a)$: 
\begin{equation}\label{local-nonholonomic-bracket}
\begin{split}
\{p_a, p_b\}_{{\mathcal D}^*}&=-C^c_{ab}(q)p_c-C^{\gamma}_{ab}(q)\eta_{\gamma}(q),\\
\{p_a, q^i\}_{{\mathcal D}^*}&=-\rho^i_a(q),\\
\{q^i, q^j\}_{{\mathcal D}^*}&=0.
\end{split}
\end{equation}
The corresponding  $2$-vector $\Pi_{{\mathcal D}^*}$ on $\D^*$ is hence given by 
\[
\Pi_{{\mathcal D}^*}=\rho^j_a\frac{\partial}{\partial q^j}\wedge \frac{\partial}{\partial p_a}
-\frac{1}{2}\left(C^c_{ab}p_c+
C^{\gamma}_{ab}\eta_{\gamma}\right)\frac{\partial}{\partial p_a}\wedge \frac{\partial}{\partial p_b}.
\]
Or, in matrix form
\[
\Pi_{{\mathcal D}^*}=
\left(
\begin{array}{cc}
\{q^i, q^j\}_{{\mathcal D}^*}&\{q^i, p_b\}_{{\mathcal D}^*}\\
\{p_a, q^j\}_{{\mathcal D}^*}&\{p_a, p_b\}_{{\mathcal D}^*}
\end{array}
\right)=
\left(
\begin{array}{cc}
0&\rho^i_b\\
-\rho^j_a&-C^c_{ab}p_c
-C^{\gamma}_{ab}\eta_{\gamma}
\end{array}
\right).
\]
Theorem \ref{th:HamForm} is equivalent to saying that the vector field  $\XHc$ describing the nonholonomic dynamics on $\D^*$ 
satisfies $\XHc = \Pi_{{\mathcal D}^*}(\cdot, d H_c)$. This is turn implies that    $\XHc$  determines the equations of motion
\[
\left(
\begin{array}{r}
\dot{q}^i\\
\dot{p}_a
\end{array}
\right)=
\left(
\begin{array}{cc}
0&\rho^i_b\\
-\rho^j_a&-C^c_{ab}p_c
-C^{\gamma}_{ab}\eta_{\gamma}
\end{array}
\right)
\begin{pmatrix}
\dfrac{\partial H_c}{\partial q^j}\\[10pt]
\dfrac{\partial H_c}{\partial p_b}
\end{pmatrix},
\]
which is the matrix form of  equations  \eqref{nonholonomic-equations2}.

\subsection{Gauge transformations}
\label{ss:gaugeT}

As first recognised in \cite{LGN-thesis, GN10}, in order to achieve Hamiltonisation of certain nonholonomic systems
possessing symmetries and a specific kind of first integrals, it is convenient to consider modifications of the
nonholonomic bracket to guarantee that such integrals descend as Casimir functions to the reduced space.
 These modifications were termed {\em gauge transformations} in \cite{BaGa} since they can be interpreted within the 
 framework of \cite{SeveraWeinstein}. We now describe how such gauge transformations can be introduced
 in our treatment. Such transformations are non-trivial when the rank $r$ of $\D$ is $\geq 3$.
 
 \subsubsection*{Definition and main properties}
 
Our presentation is inspired by  \cite{GNMo18} where  the deformation  of the bracket
 is performed using a 3-form on $Q$. In this  reference, an explicit construction of the 3-form is given 
 to guarantee that the so-called {\em invariant gauge momenta} become Casimir functions 
 of the modified bracket after reduction.
  For our purposes, we suppose that we are given  a 3-form  $\Lambda$ on  $Q$ which we assume to vanish upon 
 contraction of vectors in $\D^\perp$.
 The {\em gauge transformation of the bracket $\{\cdot, \cdot \}_{\D^*}$ by $\Lambda$} is the bracket $\{\cdot, \cdot \}^\Lambda_{\D^*}$ on $\D^*$ which is
 characterised  by 
 its value on linear and basic functions on $\D^*$ as follows:
 \begin{equation}
 \label{eq:def-nh-bracket-gauge}
\begin{split}
&\{Y^\ell, Z^\ell\}^\Lambda_{\D^*} = -(\mathcal{P}([Y, Z] + \sharp_g(i_{Y\wedge Z}\Lambda) )^\ell +( d\eta^\perp(Y, Z)
- \Lambda (Y,Z,\sharp_g (\eta^\parallel)))
 \circ \tau_{\D^*}, \\ 
 &\{Y^\ell, f \circ \tau_{\D^*}\}^\Lambda_{\D^*} = -Y(f) \circ \tau_{\D^*}, \\
 & \{f \circ \tau_{\D^*}, g \circ \tau_{\D^*}\}^\Lambda_{\D^*} = 0.
\end{split}
\end{equation}
As usual, in the above expressions $Y,Z\in \Gamma(\D)$ and $f,g\in C^\infty(Q)$. Moreover, we employ the convention
that $\langle i_{Y\wedge Z}\Lambda, W\rangle =\langle i_Z i_{Y}\Lambda, W\rangle =\Lambda(Y,Z,W)$ for $W\in \Gamma(\D)$.

\begin{remark}
\label{rmk:Lambda}
The assumption that $\Lambda$ vanishes when contracted with vectors in $\D^\perp$ is done  for concreteness and ease of presentation
since it is only the contraction of $\Lambda$ with sections of $\D$ that is relevant.
One  could alternatively define the gauge transformation by a general 3-form on $Q$ by the same formula \eqref{eq:def-nh-bracket-gauge}. With this latter approach,
the 3-forms $\Lambda_1, \Lambda_2 \in \Omega^3(Q)$  define the same gauge transformed bracket, i.e. $\{\cdot , \cdot \}^{\Lambda_1}_{\D^*} =\{\cdot , \cdot \}^{\Lambda_2}_{\D^*}$, 
 if and only if their restrictions to  $\D$ coincide.
\end{remark}


The $\R$-linearity and  skew-symmetry of the gauged transformed bracket $\{\cdot , \cdot \}^\Lambda_{\D^*}$ readily follows
from \eqref{eq:def-nh-bracket-gauge}. The validity of the Leibniz rule can be established  using the
 local expressions \eqref{eq:LocalExpressionsGauge} given below to prove the existence of a bivector on $\D^*$
 which is uniquely characterised by \eqref{eq:def-nh-bracket-gauge}.\footnote{
Alternatively, one could define the semi-basic 2-form $\Xi$ on $\D^*$ as the contraction of 
the 3-form on $\D^*$,  $(\tau_{\D^*})^*\Lambda$, with $\XHc$ (or any other second order vector field on $\D^*$) 
and proceed in analogy with \cite{GNMo18} (or other equivalent formulations which rely on 2-forms
to deform the bracket \cite{BaGa,Bal2014,Bal2017,BalYap}).}

In analogy with Theorem \ref{nonholonomic-bracket} we have the following theorem whose proof is analogous and 
relies on \eqref{eq:def-nh-bracket-gauge}.
\begin{thm}\label{nonholonomic-bracketgauge}
The bracket $\{\cdot, \cdot\}^\Lambda_{\D^*}$  satisfies: 
\begin{enumerate}
    \item[(i)] it is ${\mathbb R}$-bilinear,  skew-symmetric and satisfies the Leibniz rule in each argument;
    \item[(ii)] it is a fiberwise affine bracket on ${\mathcal D}^*$ which is fiberwise linear  if and only if
    \[
    d\eta^\perp(Y, Z)
= \Lambda (Y,Z,\sharp_g (\eta^\parallel)\; , \quad \forall \,Y, Z \in {\Gamma}({\mathcal D});
    \]
      \item[(iii)] its characteristic distribution coincides with $\E$ (defined by \eqref{eq:distE});
    \item[(iv)] If ${\mathcal D}$ is non integrable then the Jacobi identity not  satisfied.
    \end{enumerate}
\end{thm}

\subsubsection*{Local expressions}

Suppose that the 3-form $\Lambda$ is expressed in terms of our coordinate convention as:
  $$\Lambda= \frac{1}{6}\lambda_{abc}\, e^a\wedge e^b\wedge e^c,$$ where  the coefficients $\lambda_{abc}\in C^\infty(Q)$ are symmetric   with respect to even 
permutations of the subindices $a,b,c$, and skew-symmetric with respect to odd permutations.  According to the intrinsic conditions \eqref{eq:def-nh-bracket-gauge}, the presence of the 3-form $\Lambda$ modifies the local expressions \eqref{local-nonholonomic-bracket}
into
\begin{equation}
\label{eq:LocalExpressionsGauge}
\begin{split}
\{p_a, p_b\}^\Lambda_{\D^*} &= -C^c_{ab}(q)p_c-C^{\gamma}_{ab}(q)\eta_{\gamma}(q) +\lambda_{abc}(q)g^{cd}(q)(p_d-\eta_d(q)), \\
\{p_a, q^i\}^\Lambda_{{\mathcal D}^*}&=-\rho^i_a(q),\\
\{q^i, q^j\}^\Lambda_{{\mathcal D}^*}&=0.
\end{split}
\end{equation}

\begin{remark}
\label{rmk:affinegauge}
The bracket $ \{p_a, p_b\}^\Lambda_{\D^*} $ in \eqref{eq:LocalExpressionsGauge} differs from $ \{p_a, p_b\}_{\D^*} $ in
\eqref{local-nonholonomic-bracket} by the presence of a term that is affine  in the momenta $p_d$. This 
explains why the matrix $\mathcal{B}(q,p)$ in Eq. \eqref{eq:IntroPiB} in the introduction was assumed to have an affine  dependence on $p$.
\end{remark}

\subsubsection*{Description of the dynamics in terms of  gauge transformed brackets}

The following theorem shows that, for any 3-form $\Lambda$ as above, the resulting gauged-transformed bracket by $\Lambda$ describes the nonholonomic dynamics.
\begin{thm}
\label{th:HamFormGauge}
The vector field $\XHc\in \mathfrak{X}(\D^*)$ describing the nonholonomic  dynamics satisfies
\[
\XHc(\varphi) = \{\varphi, H_c\}^\Lambda_{\D^*}, \; \; \mbox{ for all } \varphi \in C^\infty(\D^*).
\]
\end{thm}
\begin{proof}
In view of Theorem \ref{th:HamForm}, it suffices to check that $ \{\varphi, H_c\}^\Lambda_{\D^*}= \{\varphi, H_c\}_{\D^*}$ for 
all  $\varphi \in C^\infty(\D^*)$. By Leibniz rule, it suffices to check that 
this holds for $\varphi$ equal to the coordinate functions $q^j$, $p_a$.   Using the local expressions 
\eqref{local-nonholonomic-bracket} and \eqref{eq:LocalExpressionsGauge} we have
\begin{equation*}
\{q^i, H_c\}^\Lambda_{\D^*} = \frac{\partial H_c}{\partial q^j}\{q^i, q^j\}^\Lambda_{\D^*}+ \frac{\partial H_c}{\partial p_a}\{q^i, p_a\}^\Lambda_{\D^*}= \frac{\partial H_c}{\partial q^j}\{q^i, q^j\}_{\D^*}+ \frac{\partial H_c}{\partial p_a}\{q^i, p_a\}_{\D^*}=\{q^i, H_c\}_{\D^*}.
\end{equation*}
On the other hand
\begin{equation*}
\{p_a, H_c\}^\Lambda_{\D^*} = \frac{\partial H_c}{\partial q^j}\{p_a, q^j\}^\Lambda_{\D^*}+ \frac{\partial H_c}{\partial p_b}\{p_a, p_b\}^\Lambda_{\D^*}= \frac{\partial H_c}{\partial q^j}\{p_a, q^j\}_{\D^*}+ \frac{\partial H_c}{\partial p_b}\{p_a, p_b\}^\Lambda_{\D^*}.
\end{equation*}
Now we claim that 
\begin{equation*}
\frac{\partial H_c}{\partial p_b}\{p_a, p_b\}^\Lambda_{\D^*}=\frac{\partial H_c}{\partial p_b}\{p_a, p_b\}_{\D^*}.
\end{equation*}
Assuming that this identity holds we obtain $ \{p_a, H_c\}^\Lambda_{\D^*}=\{p_a, H_c\}_{\D^*}$ as required. To show that the above identity holds
we use the local expression \eqref{Local-Expression-H-c} for $H_c$ to obtain $\frac{\partial H_c}{\partial p_b}=g^{be}(p_e-\eta_e)$. Therefore
\begin{equation*}
\frac{\partial H_c}{\partial p_b}\{p_a, p_b\}^\Lambda_{\D^*}=\frac{\partial H_c}{\partial p_b} \{p_a, p_b\}_{\D^*}
+ \lambda_{abd} g^{cd}(p_c-\eta_c)g^{be}(p_e-\eta_e).
\end{equation*}
But $\lambda_{abd} g^{cd}(p_c-\eta_c)g^{be}(p_e-\eta_e)=0$ for all $a$ because of the skew-symmetry  $\lambda_{abd}=-\lambda_{adb}$.
\end{proof}

\section{Almost-Poisson reduction in the presence of \\symmetries}
\label{s:Symmetry}

We now consider the reduction of the system, following   the almost-Poisson formalism, in the presence of a symmetry Lie group and  
shall see that the reduced system can also be formulated in almost-Poisson form with respect to a reduced bracket possessing 
a similar structure.  To simplify the presentation we focus on the reduction of the bracket $\{\cdot , \cdot \}_{\D^*}$ (defined by \eqref{def:constrained-bracket})
and then briefly indicate how the reduction extends for the gauge transformed brackets introduced in Subsection \ref{ss:gaugeT}.
The interest in the construction 
is that, as occurs in some examples like the Suslov problem with gyrostat treated below, the reduced system may actually be truly Poisson (perhaps after a time reparametrisation) leading to a Hamiltonisation of the system.

Throughout this section
we assume that the Lie group $G$ acts freely and properly on the configuration space $Q$ and the lifted action to $TQ$ preserves the 
Lagrangian $L$ and the distribution $\D$. We also work with the identifications \eqref{eq:D^*id} which use the kinetic energy metric 
to interpret $\D^*$ and $(\D^\perp)^*$ as
complementary subbundles of $T^*Q$.

\subsection{Reduction}
\label{ss:Red}

A crucial observation is that the invariance of $L$ implies that the action  preserves each of the 3 terms comprising the Lagrangian.
Specifically, $G$ acts by isometries on $Q$, and the gyroscopic 1-form $\eta$ and the potential $V$ are invariant. This can be concluded 
by a comparison of their homogeneity degrees as functions on $TQ$.

\begin{thm}
\label{th:reduction}
The cotangent lifted action of $G$ on $T^*Q$ leaves $\D^*$ invariant and the restricted action on $\D^*$ is  free and proper and
 leaves the nonholonomic vector field  $\XHc$, the constrained Hamiltonian 
$H_c$ and the bracket $\{\cdot , \cdot \}_{\D^*}$ invariant.
\end{thm}
\begin{proof}
Since $G$ acts by isometries,  the metric isomorphism $\flat_g:TQ\to T^*Q$ is equivariant with 
respect to the tangent and cotangent lifted actions of $G$  on $TQ$ and $T^*Q$.
Thus, the invariance of $\D$ implies the invariance of $\D^*=\flat_g(\D)$. Now, the cotangent lifted action of $G$ on 
$T^*Q$ is free and proper, so considering that $\D^*\subseteq T^*Q$ is a submanifold, the same is true about its
restriction to $\D^*$. 

The proof of invariance of the different objects on $\D^*$ proceeds by using the invariance of the metric, of the 
gyroscopic 1-form $\eta$ and of the potential $V$ to recognise the appropriate invariance/equivariance of their defining elements. For example,
the invariance of the  bracket $\{\cdot , \cdot \}_{\D^*}$ follows from the well-known invariance of the 
canonical Poisson bracket on $T^*Q$ by  cotangent lifted actions, and the equivariance of the bundle morphisms \eqref{eq:bundle-morphisms}
by the action of $G$ on $T^*Q$ and $\D^*$.
We omit the details.
\end{proof}

We now outline the reduction of the system based on the above theorem.
Since the action of $G$ on $\D^*$ is free and proper, the orbit space $\widehat{\D^*}:=\D^*/G$ is smooth and the 
orbit projection map $\pi_{\D^*}:\D^* \to \widehat{\D^*}$ is a surjective submersion. As a consequence, the bracket $\{\cdot , \cdot \}_{\D^*}$ descends
to a (smooth) {\em reduced bracket} of functions on  $\widehat{\D^*}$ characterised by
\begin{equation}
\label{eq:def-red-brack}
\{  \varphi ,  \psi \}_{\widehat{\D^*}}\circ \pi_{\D^*} = \{  \varphi\circ \pi_{\D^*} ,  \psi \circ \pi_{\D^*} \}_{\D^*}, \qquad   \varphi ,  \psi\in C^\infty(\widehat{\D^*}).
\end{equation}
 
 On the other hand, the invariance of $H_c\in C^\infty({\D^*})$ induces a {\em reduced constrained Hamiltonian} $\widehat{H_c}\in C^\infty(\widehat{\D^*})$
such that $\widehat{H_c}\circ \pi_{\D^*} =H_c$. Finally, the invariance of the nonholonomic vector field  $\XHc$ implies the existence of a {\em reduced
vector field} $\widehat{\XHc}\in \frak{X}(\widehat{\D^*})$ characterised by the property that 
 $\XHc$ and  $\widehat{\XHc}$ are  $\pi_{\D^*}$-related. Using Theorem \ref{th:HamForm} and the
expression  \eqref{eq:def-red-brack} for the reduced bracket it is easy to show that
\begin{equation}
\label{eq:def-red-brack-2}
\widehat{\XHc}(\varphi) =\{  \varphi ,  \widehat{H_c} \}_{\widehat{\D^*}}, \qquad \mbox{for all $\varphi \in  C^\infty(\widehat{\D^*})$}.
\end{equation}
In other words, the reduced dynamics may be formulated in terms of the reduced constrained Hamiltonian $\widehat{H_c}$ and the
reduced bracket on $\widehat{\D^*}$.

\subsection{Structure of the reduced space and the reduced bracket}
\label{ss:str-red-brack}

Recall that our phase space manifold $\D^*$ for the nonholonomic vector field  $\XHc$ is a rank $r$ vector bundle over the configuration
manifold $Q$. Given that  the $G$-action on $\D^*$ arises as a lift of the free and proper action of $G$ on  $Q$,
the reduced space $\widehat{\D^*}$ has  a vector bundle structure over the {\em shape space} manifold $\widehat Q:=Q/G$.
The rank of $\widehat{\D^*}$ is again $r$ and the 
corresponding bundle projection will be denoted by $\tau_{\widehat{\D^*}} : \widehat{\D^*}\to \widehat Q$. 

We  shall now give a characterisation of the reduced bracket on $\widehat{\D^*}$ in terms of linear and basic functions  in analogy 
with Proposition \ref{properties-constrained-bracket}.  For this it is useful to identify the space of smooth functions on $\widehat{\D^*}$ 
 with the space of smooth invariant  functions on $\D^*$.
In particular, linear functions on $\widehat{\D^*}$ are identified with  linear invariant functions on $\D^*$ which in turn correspond to 
elements of $\Gamma(\D)^G$, the space of $G$-invariant sections of $\D$. Basic functions 
on $\widehat{\D^*}$ are in correspondence with $C^\infty(\widehat Q)$ which is naturally identified with $C^\infty(Q)^G$, the
space of $G$-invariant functions on $Q$. On the other hand, we remark that if $Y\in \Gamma(\D)^G$ then using the $G$-invariance of $Y$ we deduce that $Y$ is $\tau_D$-projectable over a vector field on $\widehat Q$, which we denote by  $\widehat Y$.

\begin{prop}\label{properties-constrained-reduced-bracket}
If $Y,Z \in \Gamma(\D)^G$ and $\widehat{f},\widehat{g} \in C^\infty(\widehat Q)$ then:
\begin{equation*}
\begin{split}
&\{Y^\ell, Z^\ell\}_{\widehat\D^*} = -(\mathcal{P}[Y, Z])^\ell + d\eta^\perp(Y, Z) \circ \tau_{\D^*}, \\
 & \{Y^\ell, \widehat{f} \circ \tau_{\widehat{\D^*}}\}_{\widehat{\D^*}} = -\widehat{Y}(\widehat{f}) \circ \tau_{\widehat{\D^*}}, \\
 &  \{\widehat{f} \circ \tau_{\widehat{\D^*}}, \widehat{g} \circ \tau_{\widehat{\D^*}}\}_{\widehat{\D^*}} = 0.
\end{split}
\end{equation*}
\end{prop}
The proof follows almost tautologically  from our identifications and from \eqref{eq:def-red-brack} and Proposition \ref{properties-constrained-bracket}. 
The only delicate point is to interpret the elements on the right hand side of the  first equality as follows: $(\mathcal{P}[Y, Z])^\ell $ 
is  a linear function on $\widehat{ \D^*}$ since $\mathcal{P}[Y, Z]\in  \Gamma(\D)^G$. Similarly, $d\eta^\perp(Y, Z) \circ \tau_{\D^*}$ may be interpreted
as a basic function on $\widehat{ \D^*}$ since $ d\eta^\perp(Y, Z) \in C^\infty(Q)^G$.

Finally, in analogy with Theorem \ref{nonholonomic-bracket} we have the following result which is now obvious.
\begin{thm}\label{nonholonomic-bracket-reduced}
The reduced bracket $\{\cdot  , \cdot \}_{\widehat{\D^*}}$ satisfies: 
\begin{enumerate}
    \item[(i)] it is ${\mathbb R}$-bilinear, skew-symmetric and satisfies the Leibniz rule in each argument;
    \item[(ii)] it     is a fiberwise affine bracket on $\widehat{{\mathcal D}^*}$ which  is fiberwise linear  if and only if
    \[
    d\eta^{\perp} (Y, Z)=0\; , \quad \forall Y, Z \in {\Gamma}(\D)^G.
    \]
\end{enumerate}
\end{thm}
\subsection{Local expressions for the reduced bracket}
\label{ss:localreduction}

Local expressions for the reduced bracket on  $\widehat{\D^*}$ can be conveniently obtained by working with an adapted 
basis of sections  for $TQ$ whose elements are invariant. More precisely, if  the basis of sections $\{e_i\}=\{e_a, e_{\gamma}\}$ of $TQ$  is  chosen such that
all elements of the basis $\{e_a\}$ of $\D$ and $\{ e_{\gamma}\}$ of $\D^\perp$ are $G$-invariant, then, in view of the exposition
in Subsection \ref{ss:str-red-brack} above,     
the quasi-momenta $(p_a)$ are invariant fiberwise linear functions  on   $\D^*$ which may be interpreted as
fiberwise linear functions on $\widehat{\D^*}$. If $(s^I)$ are coordinates\footnote{The super-indices $I,J$ for the $s$ coordinates run from $1$ to $n-\dim G$.}   on 
the shape space $\widehat Q$, then $(s^I,p_a)$ may be used as coordinates for $\widehat{\D^*}$. It is easily verified using the discussion
above that in such coordinates we have 
\begin{equation}
\label{eq:local-red-brack}
\begin{split}
\{p_a, p_b\}_{\widehat{{\mathcal D}^*}}&=-C^c_{ab}(s)p_c-C^{\gamma}_{ab}(s)\eta_{\gamma}(s),\\
\{p_a,s^I \}_{\widehat{{\mathcal D}^*}}&=-\widehat {\rho}^{\, I}_a(s)\\
\{s^I, s^J\}_{\widehat{{\mathcal D}^*}}&=0\, .
\end{split}
\end{equation}
In the above equalities the coefficients $C^c_{ab}$, $C^{\gamma}_{ab}$ are defined as in \eqref{eq:Cabc} but are 
expressed as functions of $s$ due to their invariance. On the other hand, the coefficients $\widehat {\rho}^{\, I}_a(s)$ are determined by 
the conditions 
\[
\widehat{e_a}(s) =\widehat {\rho}^{\, I}_a(s) \frac{\partial}{\partial s^{I}}, \qquad a=1,\dots, r,
\]
where $\widehat{e_a}$ is the local vector field on $\widehat Q$ obtained as the projection of $e_a$. 

\subsection{Reduction of gauge transformed brackets}

Assuming that the 3-form $\Lambda$   is invariant under the $G$ action on $Q$,  one can extend Theorem \ref{th:reduction} to say that
the  gauged transformed bracket $\{ \cdot , \cdot \}_{\D^*}^\Lambda$ (introduced 
in Subsection \ref{ss:gaugeT}) is invariant under the  $G$ action on $\D^*$. By the same considerations given after Theorem \ref{th:reduction}, one
concludes the existence of a bracket  $\{ \cdot , \cdot \}_{\widehat{\D^*}}^\Lambda$ of functions on  $\widehat{\D^*}$ characterised by
\begin{equation*}
\{  \varphi ,  \psi \}^\Lambda_{\widehat{\D^*}}\circ \pi_{\D^*} = \{  \varphi\circ \pi_{\D^*} ,  \psi \circ \pi_{\D^*} \}^\Lambda_{\D^*}, \qquad   \varphi ,  \psi\in C^\infty(\widehat{\D^*}),
\end{equation*}
such that the reduced dynamics can be formulated in terms of this bracket with respect to the reduced constrained Hamiltonian $\widehat{H_c}$, namely,
\begin{equation*}
\widehat{\XHc}(\varphi) =\{  \varphi ,  \widehat{H_c} \}^\Lambda_{\widehat{\D^*}}, \qquad \mbox{for all $\varphi \in  C^\infty(\widehat{\D^*})$}.
\end{equation*}

Similarly, one may formulate  natural modifications
of Proposition \ref{properties-constrained-reduced-bracket}, Theorem \ref{nonholonomic-bracket-reduced} and Eqs. \eqref{eq:local-red-brack}
about the properties of  $\{ \cdot , \cdot \}_{\widehat{\D^*}}^\Lambda$ which follow, respectively, from Eqs. \eqref{eq:def-nh-bracket-gauge},
Theorem \ref{nonholonomic-bracketgauge} and Eqs. \eqref{eq:LocalExpressionsGauge}. We omit the details.

\subsection{Hamiltonisation}

A remarkable property of certain examples of symmetric nonholonomic systems is that, perhaps after a time reparametrisation,
 the reduced equations of motion possess a 
true Hamiltonian structure.  In such a case we say that the system allows a {\em Hamiltonisation}. 
This exceptional property allows the application of powerful methods of Hamiltonian systems to investigate the reduced
dynamics. Moreover, Hamiltonisable systems are natural candidates for the exploration of the dynamics 
of more general nonholonomic systems by perturbation methods.

In our context, Hamiltonisation  is achieved if one can find a  3-form $\Lambda$ and a positive  function $f\in C^\infty(\widehat{\D^*})$
such that
the reduced bracket $f \{ \cdot , \cdot \}_{\widehat{\D^*}}^\Lambda$ satisfies the Jacobi identity. In examples, the multiplicative function 
$f$ is  often basic,  i.e. $f\in C^\infty(Q)^G$, it  is
related to the density of a smooth  invariant measure of the system and is interpreted as a time reparametrisation of the dynamics. As mentioned before, for mechanical Lagrangians,  the incorporation
of the 3-form $\Lambda$ is necessary to guarantee that certain invariant first integrals descend to the quotient space as Casimirs \cite{GNMo18} and to give a geometric framework for Hamiltonisation \cite{Bal2014,Bal2017,BalYap}.

We do not present any general results about conditions for  Hamiltonisation but only some examples below.

\section{Examples: description, reduced equations and Ha\-miltonisation}
 \label{S:examples-description}
 
 We consider the Suslov and Chaplygin sphere problems  with a gyrostat and  present 
 the reduced equations and their Hamiltonian structure.
 Most of the material in this section can be found in previous references (e.g.  \cite{BorKilMamaevSuslov,MacPrz-GyrostaticSuslov} for the Suslov problem and e.g. \cite{BorMam2008,Bor-Mam-Tsiganov2014, DragovicGajicJova2023} for the Chaplygin sphere). Our goal in this section 
 is to introduce
 the necessary concepts and notation to state Theorems
 \ref{th:SuslovThm} and \ref{th:ChapBall-Thm} which clarify the geometric origin of the given  Hamiltonisation in terms of  our formalism.

\subsection{Kinematics of rigid bodies}
We briefly recall the standard approach to describe the 
motion of a rigid body in space and introduce our notation. The  orientation of the
body and  at any instant $t$ is determined by an attitude matrix $R(t)\in SO(3)$. 
Such matrix is the change of basis matrix between
an {\em inertial frame} $\{e_1,e_2, e_3\}$ and a {\em body frame} $\{E_1,E_2, E_3\}$ which is
rigidly attached to the body and rotates with it. Both the inertial 
and the body frame are assumed to be orthonormal. The {\em angular velocity vector} 
has coordinates $\omega:=(\omega_1,\omega_2,\omega_3)\in \R^3$ with respect to the inertial frame 
$\{e_1,e_2, e_3\}$ and  $\Omega:=(\Omega_1,\Omega_2,\Omega_3)\in \R^3$ with respect to the body frame $\{E_1,E_2, E_3\}$. The
coordinate vectors $\omega$ and $\Omega$ respectively correspond to the right and left 
trivialisations of the velocity vector $\dot R\in T_RSO(3)$, namely:
\begin{equation*}
\hat \Omega = R^{-1}\dot R, \qquad \hat \omega =\dot R R^{-1}.
\end{equation*}
In the above equation, we have denoted by  $\; \hat{}:  \R^3\to  \so(3)$ the standard
Lie algebra isomorphism, that, to each $v\in \R^3$ associates the skew-symmetric matrix  
$\hat v\in \so(3)$  characterised by the condition that
 $\hat v u =v\times u$ for all $u\in \R^3$ (here, 
and in what follows, `$\times$' denotes the standard vector product in $\R^3$).
It follows from our definitions and conventions that $\omega= R \Omega$.
The {\em angular momentum vector} is most naturally expressed 
 in the body frame as $M:=\I \Omega\in \R^3$ where
the {\em inertia tensor} $\I$ is a $3\times 3$ symmetric, positive definite matrix which
encodes the mass distribution of the body.
The kinetic energy of rotation of the body is given by $\frac{1}{2}M\cdot \Omega$,
where `$\cdot$' is the standard scalar product in $\R^3$. We refer the reader to 
\cite{MaRa1} for more details. 

\subsection{Effect of a gyrostat} 

%

The  total angular momentum of the body-gyrostat system is the sum of the 
body angular momentum $M=\I\Omega$ and constant moment of the gyrostat that we denote by  $B=(B_1,B_2,B_3)\in \R^3$.
The effect of the gyrostat is thus 
modelled by adding the linear gyroscopic term $B \cdot \Omega$ to the mechanical Lagrangian, namely,
\begin{equation*}
L:TQ\to \R, \qquad L=L_{mech}+B \cdot \Omega,
\end{equation*}
where $Q$ is the configuration space of the system\footnote{$Q=SO(3)$ for the Suslov  problem and $Q=SO(3)\times \R^2$ for the
Chaplygin sphere.} and $L_{mech}$ is the Lagrangian of the system in the absence of the gyrostat (in our case $L_{mech}=\mathcal{K}$, the kinetic energy).
 
\begin{remark}
The above formula for the Lagrangian can be obtained following the construction outlined in Section \ref{SS:Routh-intro}.
\end{remark}

\subsection{Suslov problem with gyrostat}
\label{ss:suslov1}

The classical  Suslov  problem concerns the rotational motion of a  rigid body about a point that
is fixed in space, under its own inertia, and subjected to the nonholonomic constraint 
that the component of the angular velocity
along an axis that is {\em fixed in the body} vanishes. The 
configuration space is  $Q=SO(3)$, and, if the body frame is chosen such
that the $E_3=(0,0,1)$  is aligned with the axis of forbidden rotation,  the nonholonomic constraint
is 
\begin{equation}
\label{eq:SuslovConstraint}
\Omega_3=0.
\end{equation}
This constraint defines a rank 2 left invariant distribution $\D$ on $Q$.  The Lagrangian
is also invariant under the (lifted) left action of $SO(3)$ and hence  the system  may be 
reduced by this action  as explained in Subsection \ref{ss:Red}.  The shape space $Q/SO(3)$
is trivial and hence the reduced phase space $\widehat{\D^*}$ is
isomorphic to $\R^2$.

The reduced equations
of motion in the presence of the gyrostat may be obtained as the restriction of the following system on $\R^3$ (see e.g.  \cite{BorKilMamaevSuslov,MacPrz-GyrostaticSuslov}):
\begin{equation}
\label{eq:SuslovR^3}
 \dot M= (M +B)\times \Omega + \lambda E_3, \qquad  \lambda = - \frac{(\I^{-1} E_3) \cdot ( (M + B)\times \Omega)}{  (\I^{-1}E_3)\cdot E_3  }.
\end{equation}
 to the  
invariant  plane $\Omega_3=0$ (it is readily checked that $\Omega_3$ is a first integral). An explicit form of these equations for  $(\Omega_1,\Omega_2)$, which may be taken as coordinates
in $\widehat{\D^*}=\R^2$,  is given below.

Due to our supposition  that $E_3$ is aligned with the axis of forbidden rotation, we may not assume
that the inertia tensor $\I$ is diagonal. However, by  a suitable rotation of the $E_1$-$E_2$ plane, we
may  assume, without loss of generality, that 
\begin{equation}
\label{eq:inertiaSuslov}
\I=\begin{pmatrix} I_{11} & 0 & I_{13} \\  0& I_{22} & I_{23} \\   I_{13}& I_{23} & I_{33} \end{pmatrix}.
\end{equation}
Elementary algebraic manipulations show that the restriction of \eqref{eq:SuslovR^3} to the
 set $\Omega_3=0$ gives the following evolution equations for $(\Omega_1,\Omega_2)\in \R^2$:
 \begin{equation}
\label{eq:SuslovR^2}
\begin{split}
	&\dot{\Omega}_1=-\frac{1}{I_{11}} \left (  I_{13} \Omega_1+ I_{23} \Omega_2+B_3\right ) \Omega_2,
	\\&\dot{\Omega}_2=\frac{1}{I_{22}} \left (  I_{13} \Omega_1+ I_{23} \Omega_2+B_3\right )
	\Omega_1.
\end{split}
\end{equation}
As predicted by the theory,  equations \eqref{eq:SuslovR^2} preserve the (constrained) energy
 \begin{equation}
 \label{eq:HSuslov}
H(\Omega_1,\Omega_2)= \frac{1}{2}\left ( I_{11}\Omega_1^2 + I_{22}\Omega_2^2 \right  ).
\end{equation}
Our treatment in Subsection \ref{ss:Suslov-geometric} will show that $H$
as given above is  the reduced constrained  Hamiltonian $\widehat{H_c}\in C^\infty(\widehat{\D^*})$
expressed in the coordinates $(\Omega_1,\Omega_2)$.

\subsubsection*{Hamiltonisation}

 As is easily verified, equations  \eqref{eq:SuslovR^2} are (Poisson) Hamiltonian  with respect to the above Hamiltonian 
function $H$ and the  bracket in $\R^2$ determined by 
\begin{equation}
\label{eq:SuslovR^2-bracket}
\{\Omega_1,\Omega_2\}: =-\frac{I_{13} \Omega_1+ I_{23} \Omega_2+B_3}{I_{11}I_{22}}.
\end{equation}
The above bracket  satisfies the Jacobi identity (any bivector on a 2-dimensional manifold does) and thus provides a {\em true} Hamiltonian formulation 
of the reduced system \eqref{eq:SuslovR^2}.

\subsubsection*{Geometric origin of the bracket \eqref{eq:SuslovR^2}}
The following theorem, whose proof is postponed to  Subsection  \ref{ss:Suslov-geometric}, explains the existence of this bracket may be understood within our theoretical framework.
 \begin{thm}
 \label{th:SuslovThm}
 The bracket \eqref{eq:SuslovR^2-bracket} coincides with the reduced bracket $\{ \cdot , \cdot \}_{\widehat{\D}^*}$.  
  \end{thm}

\subsection{Chaplygin sphere with gyrostat} 
\label{ss:chap_sph}

The classical Chaplygin sphere problem \cite{ChapBall} is one of the most celebrated examples of integrable 
nonholonomic systems and 
concerns the motion of an inhomogeneous sphere, whose center of mass coincides with its geometric center, that rolls without 
slipping on a horizontal plane due to its own inertia. The system remains integrable in the presence of a  gyrostat,
which is a generalisation of the problem due to Markeev \cite{Markeev}.%

The configuration space   is  the 5 dimensional manifold $Q=SO(3)\times \R^2$. Elements in $Q$  specify the
attitude matrix of the sphere and its position on the plane on which it rolls. There are two rolling constraints
which specify a rank 3 constraint distribution $\D$. 
The system possesses an $SE(2)$ symmetry, which satisfies the assumptions of Section \ref{s:Symmetry},
 corresponding to translations and horizontal rotations of the inertial frame on the plane on which 
the rolling takes place.

The reduced equations of motion are conveniently formulated in terms of the vector $K\in \R^3$ which gives body coordinates of the 
 angular momentum of the sphere about the contact point, and the Poisson vector $\Gamma :=R^{-1}e_3 \in \R^3$. One may explicitly write $K$ in terms of $\Omega$ and $\Gamma$ as:
 \begin{equation}
 \label{eq:Chap-sphere-AngMomContPt}
K=\I\Omega +mr^2 \Gamma \times (\Omega \times \Gamma),
\end{equation}
where $m$ is the total mass and $r$ is the radius of the sphere. One may  assume that the body frame is chosen
aligned with the principal axes of inertia and hence $\I$ is diagonal $\I=\mbox{diag}(I_1,I_2,I_3)$.  
Equation \eqref{eq:Chap-sphere-AngMomContPt}
expresses $K$ as a $\Gamma$-dependent linear function of $\Omega$. Such linear function may be inverted to give 
\begin{equation}
\label{eq:Chap-sphere-AngVelKGamma}
\Omega=AK +\left ( \frac{ (AK )\cdot \Gamma}{1-mr^2(A\Gamma)\cdot \Gamma}\right )\, A \Gamma,
\end{equation}
where $A:=(\I+mr^2\mbox{I}_3)^{-1}$, with $\mbox{I}_3$ denoting the $3\times 3$ identity matrix. We note that the denominator
\begin{equation*}
1-mr^2(A\Gamma)\cdot \Gamma>0.
\end{equation*}
Indeed, this follows for the fact that the eigenvalues of $A$ are strictly less than $\frac{1}{mr^2}$ so we may estimate $mr^2(A\Gamma)\cdot \Gamma<\|\Gamma\|^2=1$.

The $SE(2)$-reduced equations of motion are the restriction of  following system for $(K,\Gamma)\in \R^3\times \R^3$:
\begin{equation}
\label{eq:ChapBallEqns}
 \dot K =  (K +B)\times \Omega, \qquad  \dot \Gamma = \Gamma \times \Omega,
\end{equation}
to the invariant set $\|\Gamma\|^2=1$ (see e.g. \cite{BorMam2008, DragovicGajicJova2023}). In the above equations
 and in what follows, it is understood that $\Omega$ is expressed in terms of $(K,\Gamma)$ as in \eqref{eq:Chap-sphere-AngVelKGamma}. 
 
 As explained in 
Subsection \ref{ss:str-red-brack}, the reduced space $\widehat{\D^*}=\D^*/SE(2)$ is a rank 3 vector bundle over the shape space $Q/SE(2)=S^2$.
The components of $K$ are linear coordinates on the fibres of  $\widehat{\D^*}$ whereas the Poisson vector $\Gamma\in S^2$
parametrises the shape space $S^2$. Therefore we write
\begin{equation}
\label{eq:red-space-ChapSphere}
\widehat{\D^*}= \left \{ (K,\Gamma) \in  \R^3\times \R^3 \, :\, \|\Gamma\|^2=1 \, \right \},
\end{equation}
and interpret \eqref{eq:ChapBallEqns} as the equations determining the reduced vector field $\widehat{\XHc}$.

Apart from the energy
integral
\begin{equation}
H(K,\Gamma)=\frac{1}{2}K\cdot \Omega,
\end{equation}
which coincides with the reduced constrained Hamiltonian $\widehat{H_c}$, the system possesses the momentum integrals
\begin{equation}
\label{eq:integrals-ChapSphere}
F_1(K,\Gamma)=\|K+B\|^2, \qquad F_2(K,\Gamma)=(K+B)\cdot \Gamma,
\end{equation}
and the invariant measure
\begin{equation*}
\mu=(1-mr^2(A\Gamma)\cdot \Gamma)^{-1/2} \, dK d\Gamma.
\end{equation*}

Equations \eqref{eq:ChapBallEqns} can be formulated in almost-Poisson form $\dot x = \{x, H\}$, where $x=(K,\Gamma)\in \R^6$, with respect to the
bracket
\begin{equation}
\label{eq:ChapBall-bracket}
	\begin{split}
		&\{ K_i , K_j \}:=-\sum_{l=1}^{3} \varepsilon_{ijl} \left ( K_l - mr^2 (\Omega \cdot \Gamma) \Gamma_l +B_l \right ),	\\
		&\{ K_i , \Gamma_j \}:= - \sum_{l=1}^{3} \varepsilon_{ijl} \Gamma_l,
		\\&\{ \Gamma_i , \Gamma_j \}:= 0, 
	\end{split}
	\end{equation}
where  $\varepsilon_{ijk}$ denotes the  alternating tensor\footnote{\label{fnote:alttensor} $ \varepsilon_{ijl}$ equals zero if two indices are equal, it equals $1$ if $(i,j,k)$ is a cyclic permutation of 
$(1,2,3)$ and $-1$ otherwise.}. The above bracket possesses two Casimir functions  
$\| \Gamma \|^2$ and $F_2$. Since $\| \Gamma \|^2$  is a Casimir, the bracket  may be restricted to the 5-dimensional reduced space $\widehat{\D^*}$ 
defined by \eqref{eq:red-space-ChapSphere} and results in a rank $4$ almost-Poisson structure.

\subsubsection*{Hamiltonisation}
A remarkable observation, which can be verified through a direct calculation with a symbolic software, is that the rescaled bracket
\begin{equation*}
\{\cdot , \cdot \}^*:= \left (1-mr^2(A\Gamma)\cdot \Gamma \right )^{1/2} \, \{\cdot , \cdot \},
\end{equation*}
satisfies the Jacobi identity and is therefore a true Poisson bracket. Therefore, one may 
write  the reduced system  \eqref{eq:ChapBall-bracket} in Hamiltonian form  $ x'=\{x,H\}^*$, where $'=\frac{d}{d\tau}$ and the new time variable $\tau$ is related to the 
original time $t$ by 
\begin{equation*}
dt=\left (1-mr^2(A\Gamma)\cdot \Gamma \right )^{1/2} \, d\tau.
\end{equation*}
This Hamiltonisation was apparently found for the first time in \cite{BorMam2008}. (In the absence of the
gyrostat the  Hamiltonisation goes back to \cite{BorMam2001}).

\subsubsection*{Geometric origin of the bracket \eqref{eq:ChapBall-bracket}}

The geometric origin of the  bracket  \eqref{eq:ChapBall-bracket} is clarified by the following theorem.

 \begin{thm}
 \label{th:ChapBall-Thm}
 The bracket \eqref{eq:ChapBall-bracket} coincides with the reduced bracket $\{ \cdot , \cdot \}^\Lambda_{\widehat{\D}^*}$ where the 3-form $\Lambda$ can be taken as the scaling by a  constant factor of the bi-invariant Cartan volume  3-form on $SO(3)$.  
  \end{thm}

The precise scaling of  $\Lambda$ mentioned in the theorem is  clarified at the end of Subsection  \ref{ss:Chapballgeometric} below.

We mention that the gauge transformation by $\Lambda$ is necessary  to make $F_2$ into  a Casimir of \eqref{eq:ChapBall-bracket}. We found the 
specific 3-form $\Lambda$ following the prescription 
of \cite{GNMo18} where  the theory is developed only  for mechanical Lagrangians.  Theorem \ref{th:ChapBall-Thm}
above suggests that there is a natural
extension of  Theorem 5.6 of \cite{GNMo18} to the gyroscopic context considered in this work. We do not attempt to develop such extension here since
it would require a study of linear (and affine) first integrals for nonholonomic systems with gyroscopic Lagrangians which is
beyond our scope.

\section{Examples: relation with theoretical framework}
\label{S:examples-bracks}

We now treat the examples described in Section \ref{S:examples-description} following our theoretical framework.
Our purpose is to prove Theorems \ref{th:SuslovThm} and \ref{th:ChapBall-Thm}
 which  explain the geometric origin of the  brackets \eqref{eq:SuslovR^2-bracket}
and \eqref{eq:ChapBall-bracket}  in terms of our formalism.

\subsection{Euler angles in $SO(3)$} In our treatment of the examples we will use Euler angles as local coordinates in $SO(3)$.
Following the {\em $x$-convention} (see e.g. \cite{MaRa1}), we write $R\in SO(3)$ as 
\begin{equation}
\label{eq:MatR}
R=\left(
\begin{array}{ccc}
 \cos \psi \cos \varphi - \cos \theta \sin \varphi \sin \psi & -\sin \psi \cos \varphi - \cos \theta \sin \varphi \cos \psi & \sin \theta \sin \varphi \\
\cos \psi \sin \varphi + \cos \theta \cos \varphi  \sin \psi  & -\sin \psi \sin \varphi + \cos \theta \cos \varphi  \cos \psi  & -\sin \theta \cos \varphi  \\
 \sin \theta \sin \psi   & \sin \theta \cos \psi  & \cos \theta  
\end{array}
\right),
\end{equation}
where the Euler angles $0<\varphi , \psi <2\pi, \, 0<\theta <\pi$. According to this convention, one  derives 
 the following expressions for the angular velocity in space coordinates $ {\omega}$,
and in body coordinates ${\Omega}$ (see e.g. \cite{MaRa1}):
\begin{equation}
\label{E:Omega-omega}
 \omega=\left (\begin{array}{c} \dot \theta \cos \varphi +\dot \psi \sin \varphi \sin \theta \\ \dot \theta \sin \varphi - \dot \psi \cos \varphi \sin \theta  \\ \dot 
\varphi + \dot \psi \cos \theta \end{array} \right ), \qquad  \Omega=\left ( \begin{array}{c} \dot \theta \cos \psi +\dot \varphi \sin \psi \sin \theta \\ -\dot \theta \sin \psi + \dot \varphi \cos \psi \sin \theta  \\ \dot 
\varphi\cos \theta + \dot \psi  \end{array} \right ).
\end{equation}

\subsection{Suslov problem with gyrostat} 
\label{ss:Suslov-geometric}

As explained in Subsection \ref{ss:suslov1}, the configuration space is the 3-dimensional 
manifold $Q=SO(3)$ and a local expression 
 Lagrangian $L:TQ\to \R$ is obtained by
substituting \eqref{E:Omega-omega} into the right hand side of
\begin{equation}
\label{eq:SuslovLagrangian}
L(\varphi,\theta, \psi,\dot \varphi,\dot\theta, \dot \psi)=\frac{1}{2}( \I \Omega)\cdot \Omega +B\cdot \Omega,
\end{equation}
where the inertia tensor $\I$ is given by \eqref{eq:inertiaSuslov} and we recall that the constant vector $B=(B_1,B_2,B_3)\in \R^3$ is the angular momentum due to the presence of the gyrostat.
The quadratic terms in the velocities in this expression for $L$ define the kinetic energy metric 
\begin{equation*}
\begin{split}
g=&\left ( \sin \psi  \left(I_{11} \sin ^2\theta  \sin \psi +I_{13} \sin 2 \theta \right)+I_{22}  \sin ^2\theta  \cos ^2\psi +I_{23} \sin 2 \theta  \cos \psi +I_{33} \cos ^2\theta  \right ) \, d\varphi \otimes d\varphi +
\\& \left ( I_{11} \cos ^2\psi  + I_{22} \sin ^2\psi  \right ) \, d\theta \otimes d\theta + I_{33}  \, d\psi \otimes d\psi + \\&
 2  \left ( (I_{11}-I_{22}) \sin \theta  \sin \psi  \cos \psi + (I_{13} \cos \psi -I_{23} \sin \psi    ) \cos \theta \right ) \, d\varphi \otimes d\theta + \\&
2 \left ( \sin \theta  (I_{13} \sin \psi +I_{23} \cos \psi )+I_{33} \cos \theta  \right ) \, d\varphi \otimes d\psi +2 \left ( I_{13} \cos \psi -I_{23} \sin \psi  \right ) \, d\theta \otimes d\psi .
\end{split}
\end{equation*}
The linear terms instead determine the 1-form
\begin{equation}
\label{eq:Suslov-eta}
\eta= \left ((B_1 \sin \psi +B_2 \cos \psi ) \sin \theta   +B_3 \cos \theta  \right ) \, d\varphi +\left (  B_1 \cos \psi -B_2 \sin \psi   \right )\, d\theta + B_3 \, d\psi.
\end{equation}
Finally, using \eqref{E:Omega-omega}, we see that the nonholonomic constraint  \eqref{eq:SuslovConstraint} writes in local coordinates as
\begin{equation}
\label{eq:Suslov-constraint}
  \dot  \varphi\cos \theta + \dot \psi =0,
\end{equation}
which determines the rank $2$ distribution ${\mathcal D}$. 

As our adapted basis \eqref{adaptedbasis}  for ${\mathfrak X}(Q)$ we take
\begin{equation}
\label{eq:adapted-basis-Suslov}
\begin{split}
	e_1&:=\frac{1}{\sin \theta}( \sin\psi \partial_\varphi+ 
	\sin \theta \cos\psi\partial_\theta -\cos\theta \sin\psi \partial_\psi ), \\
	e_2&:=\frac{1}{\sin \theta} ( \cos\psi \partial_\varphi 
	-\sin \theta \sin\psi\partial_\theta -\cos\theta \cos\psi\partial_\psi),\\
	e_3&:=-f(\psi) \partial_\varphi+( I_{11} I_{23} \sin\psi-I_{13} I_{22} \cos \psi) \sin\theta\partial_\theta+(f(\psi) \cos\theta + I_{11} I_{22} \sin\theta)\partial_\psi,
\end{split}\end{equation}
where $f(\psi)=I_{13} I_{22} \sin\psi+I_{11} I_{23} \cos \psi$. As required,  $\{e_1, e_2\}$ is a local basis of ${\mathcal D}$ and
$\{e_3\}$  of ${\mathcal D}^\perp$. Moreover, the vector fields $e_1, e_2, e_3$ are all invariant under the action of $SO(3)$ on itself by left multiplication which,
as explained in Subsection \ref{ss:localreduction},
is convenient to obtain expressions for the reduced bracket and the constrained 
reduced Hamiltonian.

The structure coefficients $\rho_\alpha^i$ are easily read from \eqref{eq:adapted-basis-Suslov}. Denoting $(\varphi,\theta,\psi):=(q^1,q^2,q^3)$, we have
\begin{equation*}
\begin{split}
&\rho_1^1=\frac{\sin\psi}{\sin \theta}, \qquad \rho_1^2=\cos\psi, \qquad \rho_1^3=-\frac{\cos\theta \sin \psi}{\sin \theta}, \\
&\rho_2^1=\frac{\cos\psi}{\sin \theta}, \qquad \rho_2^2=-\sin\psi, \qquad \rho_1^3=-\frac{\cos\theta \cos \psi}{\sin \theta}.
\end{split}
\end{equation*}

It may be checked, using \eqref{E:Omega-omega}, that the 
quasi-velocities\footnote{For notational convenience, we use sub-indices instead of super-indices for the quasi-velocities throughout this section.} 
$v_1, \, v_2, \, v_3$ determined by our frame are related to the components of the angular velocity in body coordinates ${\Omega}$ 
by
\begin{equation*}
v_1 =\Omega_1 + \frac{I_{13}}{I_{11}}\Omega_3, \qquad v_2 =\Omega_2 + \frac{I_{23}}{I_{22}}\Omega_3, \qquad v_3 =\frac{\Omega_3}{ I_{11}I_{22}\sin\theta}.
\end{equation*}
Hence, the nonholonomic constraint \eqref{eq:SuslovConstraint} is expressed as $v_3=0$ (as expected) and along the constraint space $\D$ we have $v_1 =\Omega_1$ and $v_2 =\Omega_2$.  
The constrained Lagrangian $L_c:D\to \R$ may be  then easily computed from \eqref{eq:SuslovLagrangian} by substituting $\Omega_3=0$,  $\Omega_1=v_1$ and $\Omega_2=v_2$. One obtains 
\begin{equation*}
L_c(\varphi,\theta,\psi,v_1,v_2)=\frac{1}{2}\left ( I_{11}v_1^2+I_{22}v_2^2\right )+B_1v_1+B_2v_2.
\end{equation*}
According to \eqref{eq:Legendre-local-constrained},  the quasi-momenta $p_1$,  $p_2$, are given by
\begin{equation*}
p_1=I_{11}v_1+B_1, \qquad p_2=I_{22}v_2+B_2,
\end{equation*}
and a simple calculation shows that the constrained Hamiltonian \eqref{Local-Expression-H-c}
  is
\begin{equation*}
H_c(\varphi,\theta,\psi,p_1,p_2)=\frac{1}{2}\left ( \frac{ (p_1-B_1)^2}{I_{11}} 
+\frac{ (p_2-B_2)^2}{I_{22}} \right ).
\end{equation*}
The independence of  $H_c$ of the Euler angles is due to the $SO(3)$ symmetry 
of the problem and our choice to work with the invariant basis $e_1, e_2, e_3$. In fact, 
$p_1$ and $p_2$ can be taken as
coordinates on the reduced space $\widehat{\D^*}=\R^2$ and the above expression for $H_c(\varphi,\theta,\psi,p_1,p_2)$ coincides with 
the value of $\widehat{H_c}(p_1,p_2)$. Note that along the constraint 
space we have 
\begin{equation}
\label{eq:O1O2Suslov}
\Omega_1=\frac{p_1-B_1}{I_{11}}, \qquad \Omega_2=\frac{p_2-B_2}{I_{22}},
\end{equation}
so, if we use $(\Omega_1, \Omega_2)$  as coordinates for $\widehat{\D^*}$,
we conclude that $H$ given by \eqref{eq:HSuslov} coincides with the
constrained reduced Hamiltonian $\widehat{H_c}$.

The dual basis of our frame  $\{e_1, e_2, e_3\}$ defined in  \eqref{eq:adapted-basis-Suslov} is calculated to be
\begin{equation*}
\label{eq:adapted-dual-basis-Suslov}
\begin{split}
	e^1=& \left ( \frac{I_{13}}{I_{11}}\cos\theta+\sin\theta\sin\psi \right )   \,d\varphi + \cos\psi  \,d\theta + \frac{I_{13}}{I_{11}} \, d\psi , 
	\\ e^2=& \left (  \frac{I_{23}}{I_{22}} \cos\theta + \sin\theta \cos\psi  \right )  \, d\varphi - \sin\psi  \, d\theta + \frac{I_{23}}{I_{22}} d\psi,
	\\ e^3=& \frac{1}{I_{11}I_{22}\sin\theta} \left (  \cos\theta  \, d\varphi + d\psi \right ),
\end{split}\end{equation*}
and the expression for the 1-form $\eta$  given by \eqref{eq:Suslov-eta} in this basis is 
\begin{equation*}
	\eta = B_1 \, e^1 + B_2 \, e^2  +  \left (I_{11}I_{22}B_3 - I_{13}I_{22} B_1- I_{11}I_{23} B_2 \right ) \sin\theta  \, e^3.
\end{equation*}
Therefore, 
\begin{equation*}
\eta^\perp = (I_{11}I_{22}B_3 - I_{13}I_{22} B_1- I_{11}I_{23} B_2) \sin\theta  \, e^3.
\end{equation*}

On the other hand, one computes 
\begin{equation*}
[e_1,e_2] = \frac{I_{13}}{I_{11}} e_1 + \frac{I_{23}}{I_{22}} e_2 + \frac{1}{ I_{11} I_{22}\sin \theta} e_3,
\end{equation*}
and from this expression we read that 
\begin{equation*}\
C_{12}^1 = \frac{I_{13}}{I_{11}}, \quad 
C_{12}^2 = \frac{I_{23}}{I_{22}}, \quad
C_{12}^3 = \frac{1}{ I_{11} I_{22}\sin \theta}.
\end{equation*}

In accordance with  \eqref{local-nonholonomic-bracket}, the 
bracket  $\{\cdot , \cdot \}_{\D^*}$ is given by
\begin{equation*}
\begin{split}
&\{ p_1,p_2\}_{\D^*}=-\frac{I_{13}}{I_{11}}p_1-\frac{I_{23}}{I_{22}}p_2- \frac{I_{11}I_{22}B_3 - I_{13}I_{22} B_1- I_{11}I_{23} B_2}{ I_{11} I_{22}}, \\
&\{p_1,\varphi\}_{\D^*}=\frac{\sin\psi}{\sin \theta}, \qquad  \{p_1,\theta\}_{\D^*} =\cos\psi, \qquad  \{p_1,\psi \}_{\D^*} =-\frac{\cos\theta \sin \psi}{\sin \theta}, \\
&\{p_2,\varphi\}_{\D^*}=\frac{\cos\psi}{\sin \theta}, \qquad  \{p_2,\theta\}_{\D^*} =-\sin\psi, \qquad  \{p_2,\psi \}_{\D^*} =-\frac{\cos\theta \cos \psi}{\sin \theta}.
\end{split}
\end{equation*}

As  expected from the discussion in Subsection \ref{ss:localreduction}, 
the bracket of $p_1$ and $p_2$ is independent of the Euler angles, and 
completely determines the reduced bracket $\{ \cdot , \cdot \}_{\widehat{\D^*}}$.
Using \eqref{eq:O1O2Suslov}, we have 
\begin{equation*}
\{\Omega_1,\Omega_2\}_{\widehat{\D^*}}=-\frac{I_{13} \Omega_1+ I_{23} \Omega_2+B_3}{I_{11}I_{22}},
\end{equation*}
which coincides with the bracket  
 \eqref{eq:SuslovR^2-bracket} and proves Theorem \ref{th:SuslovThm}.

 \subsection{Chaplygin sphere with gyrostat}
 \label{ss:Chapballgeometric}
 We write $(q^1,q^2,q^3,q^4,q^5):=(\varphi,\theta,\psi,x,y)$ for the local coordinates on $Q=SO(3)\times \R^2$, where  
  $(x,y)\in \R^2$  are the horizontal coordinates, with respect to the space frame,
 of the center of the  sphere on the plane.
  The  
 Lagrangian $L:TQ\to \R$ is the sum of the total kinetic energy and the
 term accounting for the gyrostat, so
 \begin{equation}
 \label{eq:LagChapSphere}
L= \frac{1}{2}(\I \Omega)\cdot \Omega +\frac{m}{2}(\dot x^2 + \dot y^2) +B\cdot \Omega,
\end{equation}
where $m$ is the total mass of the sphere. The explicit value of $L$ 
in terms of the bundle coordinates $(q^i,\dot q^i)$ 
  is obtained by
substituting \eqref{E:Omega-omega} into the right hand side of \eqref{eq:LagChapSphere}.

 Recall  that the inertia tensor in this case is assumed to be  diagonal $\I=\mbox{diag}(I_1,I_2,I_3)$. In order to simplify the exposition, we  will also assume that $I_1=I_2$. The procedure in the general case is identical but the resulting expressions are much more involved. 
The kinetic energy metric simplifies to
\begin{equation}
\label{eq:kinetic-energy-ChapSphere}
\begin{split}
g=&f(\theta) \, d\varphi \otimes d\varphi +I_1 \, d\theta \otimes d\theta+I_3 \, d\psi \otimes d\psi
 +2 I_{3} \cos \theta  \, d\varphi \otimes d\psi  + m \,( dx\otimes dx + dy\otimes dy),
\end{split}
\end{equation}
where we  denote $f(\theta):= I_1 \sin ^2\theta  +I_{3} \cos ^2\theta$.
The 1-form $\eta$ is again given by \eqref{eq:Suslov-eta}. The nonholonomic constraints
of rolling without slipping are 
\begin{equation}
\label{eq:ChapSphere-Constraintbis}
\dot x = r \omega_1, \qquad 
\dot y = -r\omega_2,
\end{equation}
 which, using \eqref{E:Omega-omega}, rewrite  in local coordinates as
\begin{equation}
\label{eq:ChapSphere-Constraint}
 \dot x -r \sin \varphi \,  \dot \theta +r \cos \varphi \sin \theta \, \dot \psi=0, \qquad 
 \dot y +r \cos \varphi \,  \dot \theta +r \sin \varphi \sin \theta \, \dot \psi=0,
\end{equation}
and determine the rank $3$ distribution ${\mathcal D}$ on $Q$. 

 Denote by $(R_\alpha, (a, b))\in SE(2) $ the matrix
\begin{equation*}
(R_\alpha, (a, b)):=\begin{pmatrix} \cos \alpha & -\sin \alpha & a \\   \sin \alpha & \cos \alpha & b \\ 0 & 0 & 1 \end{pmatrix}.
\end{equation*}
Then, the  symmetry action of  $SE(2)$ on $Q$ is is given in our coordinates by 
\begin{equation}
\label{eq:SE2action}
(R_\alpha, (a, b)): (\varphi,\theta,\psi,x,y) \mapsto (\varphi + \alpha ,\theta,\psi,x \cos \alpha - y\sin \alpha + a, x \sin \alpha + y\cos \alpha + b).
\end{equation}

As our adapted basis \eqref{adaptedbasis}  for ${\mathfrak X}(Q)$ we take
\begin{equation}
\label{eq:adapted-basis-ChapSphere}
\begin{split}
	e_1&:= \partial_\varphi, \\
	e_2&:=\partial_\theta + r\sin \varphi \, \partial_x - r\cos \varphi  \,  \partial_y \\
	e_3&:=\partial_\psi - r \sin\theta (\cos\varphi  \, \partial_x + \sin\varphi \, \partial_y), \\
	e_4&:=\frac{m r^2 \sin\theta}{I_1}\partial_\theta- r\sin\theta \sin\varphi \, \partial_x +r  \sin\theta \cos\varphi 
	\,  \partial_y,\\ e_5&: = -\frac{m r^2 \cos\theta}{I_1}\partial_\varphi+\frac{m r^2 f(\theta) }{  I_1 I_3} \partial_\psi+r \sin\theta \cos\varphi \,\partial_x+ r \sin\theta \sin\varphi\,\partial_y.
\end{split}\end{equation}
As required,  $\{e_1, e_2, e_3\}$ is a local basis of ${\mathcal D}$ and
$\{e_4, e_5\}$  of ${\mathcal D}^\perp$. Moreover, as may be verified, these vector fields are invariant
under the $SE(2)$ action \eqref{eq:SE2action}. (We note that the treatment in Section 6 of \cite{GNMo18},
which considers the Chaplygin sphere without gyrostat, 
is done  with the same basis of ${\mathcal D}$ and, therefore, some of the expressions given below 
can be found there.)

The dual basis of our frame   \eqref{eq:adapted-basis-ChapSphere} is computed to be
\begin{equation}
\label{eq:adapted-dual-basis-ChapSphere}
\begin{split}
	e^1=& d\varphi +\frac{mr^2 I_3 \cos\theta}{G(\theta)}\left (  d\psi  + \frac{\cos \varphi \, dx +\sin \varphi \, dy}{r\sin \theta}\right ),
	\\ e^2=&  \frac{1}{I_1+mr^2}\left ( I_1 \, d\theta + mr \sin \varphi \, dx-  mr \cos \varphi \, dy \right ),
	\\ e^3=&  \frac{1}{G(\theta) \sin \theta}\left ( I_1I_3 \sin \theta  \, d\psi - mr f(\theta) (\cos \varphi \, dx+ \sin \varphi \, dy )\right ), 
	\\ e^4=&  \frac{I_1}{(I_1+mr^2)\sin \theta}\left (  d\theta + \frac{1}{r} \left ( - \sin \varphi \, dx+ \cos \varphi \, dy \right ) \right ),
	\\ e^5=& \frac{I_1 I_3}{r G(\theta) \sin \theta}\left (r \sin \theta  \, d\psi +  \cos \varphi \, dx+ \sin \varphi \, dy  \right ),  
\end{split}\end{equation}
where 
\begin{equation*}
G(\theta):=I_1I_3  +I_1mr^2 \sin^2\theta +I_3mr^2 \cos^2\theta. 
\end{equation*}
The 1-form $\eta$  given by \eqref{eq:Suslov-eta} writes as $\eta= \eta_je^j$ with coefficients
\begin{equation}
\label{eq:eta-ChapSphere}
\begin{split}
	&\eta_1 =  B_1\sin \theta \sin \psi  + B_2 \sin \theta \cos \psi  +B_3 \cos \theta, \qquad \eta_2= B_1 \cos\psi - B_2 \sin\psi, \\ &\eta_3=B_3, \qquad
	\eta_4= \frac{m r^2 \sin\theta}{I_1} \left ( B_1 \cos\psi -B_2 \sin\psi  \right ), \\  &\eta_5= 
	\frac{m r^2 \sin^2\theta}{I_1I_3} \left (  I_1 B_3 -  \frac{I_3  \cos\theta}{\sin \theta} (B_1 \sin\psi + B_2 \cos\psi) \right ).
	\end{split}
\end{equation}

The coefficients $\rho_\alpha^i$ are easily read from \eqref{eq:adapted-basis-ChapSphere}. Recalling our convention $(q^1,q^2,q^3,q^4,q^5)=(\varphi,\theta,\psi,x,y)$, we have
\begin{equation}
\label{eq:rho-ChapSphere}
\begin{split}
&\rho_1^1=1, \qquad \rho_1^2=0, \qquad \rho_1^3=0, \qquad \rho_1^4=0, \qquad \rho_1^5=0, \\
&\rho_2^1=0, \qquad \rho_2^2=1, \qquad \rho_2^3=0, \qquad \rho_2^4=r\sin \varphi, \qquad 
\rho_2^5=- r\cos \varphi,
\\
&\rho_3^1=0, \qquad \rho_3^2=0, \qquad \rho_3^3=1, \qquad \rho_3^4=- r\sin \theta \cos \varphi, \qquad 
\rho_3^5=- r\sin \theta \sin \varphi.
\end{split}
\end{equation}

On the other hand, via the calculation of the pairwise brackets $[e_a,e_b]$, $1\leq a<b\leq 3$, one determines after long
but straightforward calculations:
\begin{small}
\begin{equation}
\label{eq:C_ab^j-ChapSphere}
	\begin{split}
	C_{12}^1 & =-C_{21}^1= \frac{m r^2 I_3 \cos\theta}{G(\theta)\sin\theta }, \qquad 
	 C_{12}^3  =-C_{21}^3 =-\frac{m r^2 f(\theta)}{G(\theta)\sin\theta  }, \qquad 
	 C_{12}^5  =-C_{21}^5=\frac{ I_1 I_3}{G(\theta)\sin\theta }, \\
	 C_{13}^2  & = -C_{31}^2=\frac{m r^2 \sin\theta}{I_1 + m r^2 }, \qquad 
	C_{13}^4 =-C_{31}^4=-\frac{ I_1}{I_1+ m r^2}, \qquad 
	 C_{23}^1   =-C_{32}^1 =-\frac{m r^2 I_3 \cos^2\theta}{G(\theta) \sin\theta }, \\
	 C_{23}^3 &=-C_{32}^3=\frac{m r^2 \cos\theta f(\theta)}{G(\theta) \sin\theta  }, \qquad 
	 C_{23}^5 =-C_{32}^5=-\frac{I_1 I_3 \cos\theta}{ G(\theta) \sin\theta },
	\end{split}
	\end{equation}
	\end{small}
	where we have omitted all  terms $C_{ab}^j$ which vanish.	
	In view of \eqref{local-nonholonomic-bracket},
 we obtain  the following 
	expressions for the bracket $\{\cdot , \cdot \}_{\D^*}$:
	\begin{equation*}
	\begin{split}
\{p_1,p_2\}_{\D^*}&=- \frac{m r^2 I_3 \cos\theta}{G(\theta)\sin\theta } p_1+\frac{m r^2 f(\theta)}{G(\theta)\sin\theta  } p_3-\frac{ I_1 I_3}{G(\theta)\sin\theta } \eta_5,\\
\{p_1,p_3\}_{\D^*}&=-\frac{m r^2 \sin\theta}{I_1 + m r^2 }p_2 + \frac{ I_1}{I_1+ m r^2} \eta_4, \\
\{p_2,p_3\}_{\D^*}&=\frac{m r^2 I_3 \cos^2\theta}{G(\theta) \sin\theta }p_1-\frac{m r^2 \cos\theta f(\theta)}{G(\theta) \sin\theta  }p_3+\frac{I_1 I_3 \cos\theta}{ G(\theta) \sin\theta } \eta_5,
\end{split}
\end{equation*}
with $\eta_4$ and $\eta_5$ given by \eqref{eq:eta-ChapSphere}.
The other non-vanishing brackets of the form  $\{p_a, q^i \}_{\D^*}$, are readily determined from \eqref{local-nonholonomic-bracket} and  \eqref{eq:rho-ChapSphere} to be
\begin{equation}
\label{eq:rho-ChapSphereBrack}
\begin{split}
&\{p_1,\varphi\}_{\D^*}=-1, \qquad  \{p_2,\theta\}_{\D^*}=-1,  \qquad \{p_2,x\}_{\D^*}=-r\sin \varphi, \qquad 
\{p_2,y\}_{\D^*}= r\cos \varphi,
\\
&\{p_3,\psi\}_{\D^*}=-1, \qquad \{p_3,x\}_{\D^*}= r\sin \theta \cos \varphi, \qquad 
 \{p_3,y\}_{\D^*}= r\sin \theta \sin \varphi.
\end{split}
\end{equation}

\subsubsection*{Gauge transformation}

The first integral $F_2$ defined  by \eqref{eq:integrals-ChapSphere} 
equals the quasi-momentum $p_1$ (this follows immediately from 
the formulae \eqref{eq:GammachapSphere} and \eqref{eq:KchapSphere}  given ahead). Therefore, 
inspired by \cite{GNMo18}, we consider the gauge transformation by the 3-form $\Lambda=\tilde \lambda \, e^1\wedge e^2\wedge e^3$ where $\tilde \lambda\in C^\infty(Q)$ is
given by
\begin{equation*}
\tilde \lambda := g([e_1,e_2],e_3) =-mr^2\sin \theta.
\end{equation*}
Hence, the coefficients $\lambda_{abc}$ appearing in  \eqref{eq:LocalExpressionsGauge} are given by
\begin{equation*}
\lambda_{abc} = - \varepsilon_{abc} mr^2\sin \theta,
\end{equation*}
(see footnote \ref{fnote:alttensor} for the definition of  the alternating tensor $\varepsilon_{abc}$).

According to \eqref{eq:LocalExpressionsGauge},  to  obtain the local expressions for the gauge-transformed bracket $\{\cdot , \cdot \}^\Lambda_{\D^*}$, 
we need the entries of the matrix
$g^{ab}$ which is the inverse of the $3\times 3$ matrix with entries $g_{ab}:=g(e_a,e_b)$, that determines
the quadratic part of the constrained Lagrangian $L_c$ (see Eq.  \eqref{eq:const_lag_chapsphere} below). One finds the following 
expressions for $g^{ab}$ (which are also given in \cite{GNMo18}):
\begin{equation}
\label{eq:invmetricChapSphere}
\begin{split}
g^{11}&=\frac{I_3+mr^2\sin^2 \theta}{G(\theta) \sin ^2\theta}, \qquad g^{22}=\frac{1}{I_1+mr^2}, \qquad g^{33}=\frac{f(\theta)}{G(\theta) \sin ^2\theta},\\
g^{12}&=g^{21}=0, \qquad  g^{13}=g^{31}=-\frac{I_3\cos \theta}{G(\theta) \sin ^2\theta} \qquad  g^{23}=g^{32}=0.
\end{split}
\end{equation}

The gauged transformed bracket $\{\cdot , \cdot \}^\Lambda_{\D^*}$ can be  determined substituting the above formulae into \eqref{eq:LocalExpressionsGauge}.
For instance, one computes 
\begin{equation*}
\begin{split}
\{p_1,p_2\}^\Lambda_{\D^*}&=\{p_1,p_2\}_{\D^*} +\lambda_{12a}g^{ab}(p_b-\eta_b) \\ &=\{p_1,p_2\}_{\D^*} + \lambda_{123}g^{3b}(p_b-\eta_b) \\
&=\{p_1,p_2\}_{\D^*} -mr^2\sin \theta \left  ( -\frac{I_3\cos \theta}{G(\theta) \sin ^2\theta} (p_1-\eta_1)+\frac{f(\theta)}{G(\theta) \sin ^2\theta}  (p_3-\eta_3) \right ),
\end{split}
\end{equation*}
and similarly for the other brackets. Simplifying the resulting expressions one obtains 
\begin{equation}
\label{eq:gauged-bracket}
	\begin{split}
&\{p_1,p_2\}^\Lambda_{\D^*}=0, \qquad 
\{p_1,p_3\}^\Lambda_{\D^*}=0, \\
&\{p_2,p_3\}^\Lambda_{\D^*}=\frac{mr^2\sin \theta}{G(\theta)}\left (-(I_3+mr^2)p_1 +(I_3-I_1)\cos\theta p_3 \right . \\
&\qquad \qquad  \qquad \left . +
(I_3+mr^2)\sin \theta (\sin \psi B_1+\cos \psi B_2)+(I_1+mr^2)\cos \theta B_3 \right ).
\end{split}
\end{equation}
According to \eqref{eq:LocalExpressionsGauge}, the brackets of the quasi-momenta $p_a$ with the coordinates  $q^i$ are not modified by 
the gauge transformation. Hence, the non-vanishing expressions of the form $\{p_a,q^i\}^\Lambda_{\D^*}$ 
exactly agree with those given above in \eqref{eq:rho-ChapSphereBrack}. In other words,
\begin{equation}
\label{eq:rho-ChapSphereBrackgauge}
\begin{split}
&\{p_1,\varphi\}^\Lambda_{\D^*}=-1, \qquad  \{p_2,\theta\}^\Lambda_{\D^*}=-1,  \qquad \{p_2,x\}^\Lambda_{\D^*}=-r\sin \varphi, \qquad 
\{p_2,y\}^\Lambda_{\D^*}= r\cos \varphi,
\\
&\{p_3,\psi\}^\Lambda_{\D^*}=-1, \qquad \{p_3,x\}^\Lambda_{\D^*}= r\sin \theta \cos \varphi, \qquad 
 \{p_3,y\}^\Lambda_{\D^*}= r\sin \theta \sin \varphi.
\end{split}
\end{equation}

\subsubsection*{Expressions for the bracket in terms of $K$ and $\Gamma$}

In order to prove Theorem \ref{th:ChapBall-Thm}, we compute bracket $\{ \cdot , \cdot \}^\Lambda_{\D^*}$,
determined by Eqs.  \eqref{eq:gauged-bracket} and \eqref{eq:rho-ChapSphereBrackgauge} above,
 of the entries of the Poisson vector  $\Gamma$ and the momentum $K$ introduced
in Subsection \ref{ss:chap_sph}. For this we need to express $\Gamma$ and $K$ in terms of 
our coordinates $q^j$ and quasi-momenta $p_a$.
 In view of \eqref{eq:MatR}, the entries of $\Gamma=R^{-1}e_3$ are: 
\begin{equation}
\label{eq:GammachapSphere}
\Gamma_1=\sin \theta \sin \psi, \qquad \Gamma_2=\sin \theta \cos \psi \qquad \Gamma_3=\cos \theta.
\end{equation}

On the other hand, it may be checked, using \eqref{E:Omega-omega}, that, along the distribution $\D\subset TQ$ defined
by nonholonomic constraints \eqref{eq:ChapSphere-Constraintbis},  the 
quasi-velocities
$v_1, \dots, v_5$ determined by our frame satisfy $v_4=v_5=0$ (as expected) and 
\begin{equation}
\label{eq:v-ChapSphere}
	\begin{split}
	v_1&=\frac{\sin\psi \Omega_1+\cos\psi \Omega_2}{\sin\theta}, 
	\qquad v_2=\cos\psi \Omega_1-\sin\psi \Omega_2,	
	\\ 	v_3&=\Omega_3 -\frac{\cos\theta}{\sin\theta} (\sin\psi \Omega_1+\cos\psi \Omega_2).
%
\end{split}
	\end{equation}	

Furthermore,  
the constrained Lagrangian $L_c:\D\to \R$ is given by
\begin{equation}
\begin{split}
\label{eq:const_lag_chapsphere}
L_c(\varphi,\theta,\psi,x,y, v_1,v_2,v_3) &= \frac{1}{2} f(\theta) v_1^2 + \frac{I_1 + mr^2}{2} v_2^2 + \frac{1}{2} (I_3+mr^2 \sin^2\theta)v_3^2 \\ & \qquad + I_3 \cos\theta  v_1 v_3  +\eta_1 v_1 +\eta_2 v_2 +\eta_3 v_3,
\end{split}
\end{equation}
where $\eta_1, \eta_2, \eta_3$ are given by \eqref{eq:eta-ChapSphere}. Hence, the 
quasi-momenta $p_1$,  $p_2$, $p_3$,  are given by
	\begin{equation}
	\label{eq:p-intermsofv-ChapSphere}
	\begin{split}
	p_1=& f(\theta) v_1 + I_3 \cos\theta v_3 + \eta_1,
	\\p_2=&  (I_1+mr^2)v_2+\eta_2,
	\\p_3=& I_3\cos\theta v_1 + (I_3+mr^2 \sin^2\theta)v_3 + \eta_3.
	\end{split}
	\end{equation}
Combining the above expressions shows that the entries of $K$ as defined by 
\eqref{eq:Chap-sphere-AngMomContPt} satisfy:
	\begin{equation}
	\label{eq:KchapSphere}
	\begin{split}
	K_1=&\frac{\sin\psi}{\sin\theta}p_1+\cos\psi p_2-\frac{\cos\theta\sin\psi}{\sin\theta}p_3-B_1,
	\\K_2=&\frac{\cos\psi}{\sin\theta}p_1-\sin\psi p_2-\frac{\cos\theta\cos\psi}{\sin\theta}p_3-B_2,
	\\K_3=&p_3-B_3.
	\end{split}
	\end{equation}
	
	A direct calculation (which is best to perform  with the help of a computer algebra software),
	  using the above expressions for $\Gamma$ and $K$ 
		together with  \eqref{eq:gauged-bracket} and \eqref{eq:rho-ChapSphereBrackgauge},   shows that the 
pairwise brackets $\{K_i,K_j\}^\Lambda_{\D^*}, \{\Gamma_i,\Gamma_j\}^\Lambda_{\D^*}, \{K_i,\Gamma_j\}^\Lambda_{\D^*}$
satisfy \eqref{eq:ChapBall-bracket}. For such calculations  it may be useful to note that
\begin{equation*}
G(\theta)=(I_1+mr^2)(I_3+mr^2)(1-mr^2\Gamma \cdot A\Gamma) \quad \mbox{and} \quad \{ p_2,p_3\}_{\D^*}^\Lambda=-mr^2 \sin \theta \Omega \cdot \Gamma.
\end{equation*}

Finally, we mention that,  using  expressions in Eqs \eqref{eq:invmetricChapSphere} for the matrix 
$g^{ab}$, together with 
\eqref{eq:v-ChapSphere}, \eqref{eq:p-intermsofv-ChapSphere}, \eqref{eq:KchapSphere} one may verify
that the expression \eqref{Local-Expression-H-c} for the constrained 
Hamiltonian $H_c$ (with no potential) equals $\frac{1}{2}K\cdot \Omega$. Clearly, the same 
holds true for the reduced constrained 
Hamiltonian $\widehat{H_c}$ as had been claimed in 
subsection \ref{ss:chap_sph}.

\subsubsection*{Relationship between $\Lambda$ and the Cartan bi-invariant volume form on $SO(3)$.}

The gauge transformation above was done in terms of the 3-form
\begin{equation*}
\Lambda = -mr^2 \sin \theta \, e^1\wedge e^2\wedge e^3,
\end{equation*}
with $e^1, e^2, e^3$ given by \eqref{eq:adapted-dual-basis-ChapSphere}. To complete the proof
of Theorem \ref{th:ChapBall-Thm}, it remains to show that $\Lambda$ can be also taken as a suitable scaling
of the Cartan bi-invariant volume form on $SO(3)$. For this we use that 
any  bi-invariant 3-form on $SO(3)$  is a constant multiple of  
\begin{equation*}
\nu :=-\sin \theta \, d\varphi\wedge d\theta \wedge d\psi,
\end{equation*}
(see e.g. \cite{GNMo18}). Using the above expressions for $\Lambda$ and $\nu$,
 together with  \eqref{eq:adapted-basis-ChapSphere}, it is clear that
\begin{equation*}
\Lambda (e_a,e_b,e_c)= mr^2 \nu (e_a,e_b,e_c),
\end{equation*}
for any $a,b,c\in \{1,2,3\}$. In other words, given that $\{e_1,e_2,e_3\}$ is a basis of sections of $\D$, 
we conclude that the 3-forms $\Lambda$ and $mr^2 \nu$ coincide when restricted to $\D$. Therefore, 
by Remark  \ref{rmk:Lambda}, it follows that the gauge transformed bracket by $\Lambda$ coincides with 
the gauge transformed bracket by $mr^2 \nu$.

\section*{Acknowledgements}LGN acknowledges support from the projects MIUR-PRIN 20178CJA2B {\em New Frontiers of Celestial Mechanics: theory and applications} and
MIUR-PRIN 2022FPZEES {\em Stability in Hamiltonian dynamics and beyond},  and the hospitality of 
the Department of Mathematics at Universidad de La Laguna where part of this work was done. JCM and DMdD acknowledge financial support from the Spanish Ministry of Science and Innovation under grants PID2022-137909NB-C22 and  PID2022-137909NB-C21, respectively, and  RED2022-134301-TD (MCIN/AEI/10.13039/301100011033).

\bibliographystyle{plainnat}
\bibliography{references}

\begin{thebibliography}{10}

\bibitem{AM78}
R.~Abraham and J.E. Marsden.
\newblock {\em Foundations of Mechanics}.
\newblock Benjamin/Cummings Publishing Co.Inc. Advanced Book Program, Reading,
  Mass, 1978.

\bibitem{Bal2014}
Paula Balseiro.
\newblock The {J}acobiator of nonholonomic systems and the geometry of reduced
  nonholonomic brackets.
\newblock {\em Arch. Ration. Mech. Anal.}, 214(2):453--501, 2014.

\bibitem{Bal2017}
Paula Balseiro.
\newblock Hamiltonization of solids of revolution through reduction.
\newblock {\em J. Nonlinear Sci.}, 27(6):2001--2035, 2017.

\bibitem{BalFer}
Paula Balseiro and Oscar~E. Fernandez.
\newblock Reduction of nonholonomic systems in two stages and
  {H}amiltonization.
\newblock {\em Nonlinearity}, 28(8):2873--2912, 2015.

\bibitem{BaGa}
Paula Balseiro and Luis~C. Garc\'{\i}a-Naranjo.
\newblock Gauge transformations, twisted {P}oisson brackets and
  {H}amiltonization of nonholonomic systems.
\newblock {\em Arch. Ration. Mech. Anal.}, 205(1):267--310, 2012.

\bibitem{BalYap}
Paula Balseiro and Luis~P. Yapu.
\newblock Conserved quantities and {H}amiltonization of nonholonomic systems.
\newblock {\em Ann. Inst. H. Poincar\'{e} C Anal. Non Lin\'{e}aire},
  38(1):23--60, 2021.

\bibitem{BLMMM2011}
M.~Barbero-Li{\~n}{\'a}n, M.~de~Le{\'o}n, D.~M. de~Diego, J.C. Marrero, and
  M.~C. Mu{\~n}oz-Lecanda.
\newblock Kinematic reduction and the hamilton-jacobi equation.
\newblock {\em J. Geom. Mech.}, 4(3):207--237, 2012.

\bibitem{BatesSniaticky}
Larry Bates and J~\'{S}niatycki.
\newblock Nonholonomic reduction.
\newblock {\em Rep. Math. Phys.}, 32(1):99--115, 1993.

\bibitem{Bloch}
A.~Bloch.
\newblock {\em Nonholonomic Mechanics and Control}.
\newblock Springer, Interdisciplinary Applied Mathematics 24, 2015.

\bibitem{BMZ2005}
A.~Bloch, J.E. Marsden, and D.V. Zenkov.
\newblock Nonholonomic dynamics.
\newblock {\em Notices Amer. Math. Soc.}, 52(3):324--333, 2005.

\bibitem{BKMM}
Anthony~M. Bloch, P.~S. Krishnaprasad, Jerrold~E. Marsden, and Richard~M.
  Murray.
\newblock Nonholonomic mechanical systems with symmetry.
\newblock {\em Arch. Rational Mech. Anal.}, 136(1):21--99, 1996.

\bibitem{BoMaZe}
Anthony~M. Bloch, Jerrold~E. Marsden, and Dmitry~V. Zenkov.
\newblock Quasivelocities and symmetries in non-holonomic systems.
\newblock {\em Dyn. Syst.}, 24(2):187--222, 2009.

\bibitem{Bolsinov2015}
A~V Bolsinov, A~V Borisov, and I~S Mamaev.
\newblock Geometrisation of chaplygin's reducing multiplier theorem.
\newblock {\em Nonlinearity}, 28(7):2307, jun 2015.

\bibitem{BorMam2001}
A.~V. Borisov and I.~S. Mamaev.
\newblock Chaplygin's ball rolling problem is {H}amiltonian.
\newblock {\em Mat. Zametki}, 70(5):793--795, 2001.

\bibitem{BorisovMamaev}
A.~V. Borisov and I.~S. Mamaev.
\newblock Isomorphism and the {H}amiltonian representation of some nonholonomic
  systems.
\newblock {\em Sibirsk. Mat. Zh.}, 48(1):33--45, 2007.

\bibitem{BorMam2008}
A.~V. Borisov and I.~S. Mamaev.
\newblock Conservation laws, hierarchy of dynamics and explicit integration of
  nonholonomic systems.
\newblock {\em Regul. Chaotic Dyn.}, 13(5):443--490, 2008.

\bibitem{Bor-Mam-Tsiganov2014}
A.~V. Borisov, I.~S. Mamaev, and A.~V. Tsyganov.
\newblock Nonholonomic dynamics and {P}oisson geometry.
\newblock {\em Uspekhi Mat. Nauk}, 69(3(417)):87--144, 2014.

\bibitem{BorKilMamaevSuslov}
Alexey~V. Borisov, Alexander~A. Kilin, and Ivan~S. Mamaev.
\newblock Hamiltonicity and integrability of the {S}uslov problem.
\newblock {\em Regul. Chaotic Dyn.}, 16(1-2):104--116, 2011.

\bibitem{BoKiMa}
Alexey~V. Borisov, Alexander~A. Kilin, and Ivan~S. Mamaev.
\newblock How to control {C}haplygin's sphere using rotors.
\newblock {\em Regul. Chaotic Dyn.}, 17(3-4):258--272, 2012.

\bibitem{BoKiMa2}
Alexey~V. Borisov, Alexander~A. Kilin, and Ivan~S. Mamaev.
\newblock On the {H}adamard-{H}amel problem and the dynamics of wheeled
  vehicles.
\newblock {\em Regul. Chaotic Dyn.}, 20(6):752--766, 2015.

\bibitem{BMB2015}
Alexey~V. Borisov, Ivan.~S. Mamaev, and Ivan.~A. Bizyaev.
\newblock The {J}acobi integral in nonholonomic mechanics.
\newblock {\em Regul. Chaotic Dyn.}, 20(3):383--400, 2015.

\bibitem{CaLeMaMa}
F.~Cantrijn, M.~de~Le\'{o}n, J.~C. Marrero, and D.~M. de~Diego.
\newblock Reduction of nonholonomic mechanical systems with symmetries.
\newblock volume~42, pages 25--45. 1998.
\newblock Pacific Institute of Mathematical Sciences Workshop on Nonholonomic
  Constraints in Dynamics (Calgary, AB, 1997).

\bibitem{CanLeoMar}
F.~Cantrijn, M.~de~Le\'{o}n, and D.~Mart\'{\i}n~de Diego.
\newblock On almost-{P}oisson structures in nonholonomic mechanics.
\newblock {\em Nonlinearity}, 12(3):721--737, 1999.

\bibitem{CaCoLeMa}
Frans Cantrijn, Jorge Cort\'{e}s, Manuel de~Le\'{o}n, and David Mart\'{\i}n~de
  Diego.
\newblock On the geometry of generalized {C}haplygin systems.
\newblock {\em Math. Proc. Cambridge Philos. Soc.}, 132(2):323--351, 2002.

\bibitem{ChapBall}
S.~A. Chaplygin.
\newblock On a ball's rolling on a horizontal plane.
\newblock {\em Regul. Chaotic Dyn.}, 7(2):131--148, 2002.

\bibitem{Cortes}
J.~Cort\'{e}s.
\newblock {\em Geometric, control and numerical aspects of nonholonomic
  systems}, volume 1793 of {\em Lecture Notes in Mathematics}.
\newblock Springer-Verlag, Berlin, 2002.

\bibitem{CdLMM}
Jorge Cort\'{e}s, Manuel de~Le\'{o}n, Juan~Carlos Marrero, and Eduardo
  Mart\'{\i}nez.
\newblock Nonholonomic {L}agrangian systems on {L}ie algebroids.
\newblock {\em Discrete Contin. Dyn. Syst.}, 24(2):213--271, 2009.

\bibitem{Cu-Du}
Richard Cushman, Hans Duistermaat, and J.~\'{S}niatycki.
\newblock {\em Geometry of nonholonomically constrained systems}, volume~26 of
  {\em Advanced Series in Nonlinear Dynamics}.
\newblock World Scientific Publishing Co. Pte. Ltd., Hackensack, NJ, 2010.

\bibitem{LMD2010}
M.~de~Le{\'o}n, J.C. Marrero, and D.~M. de~Diego.
\newblock Linear almost poisson structures and hamilton-jacobi equation.
  applications to nonholonomic mechanics.
\newblock {\em J. Geom. Mech.}, 2(2):159--198, 2010.

\bibitem{LR89}
M.~de~Le{\'o}n and P.R. Rodrigues.
\newblock {\em Methods of differential geometry in analytical mechanics}.
\newblock North Holland Mathematics Studies 158, 1989.

\bibitem{LeLaLoMa2023}
Manuel de~Le\'{o}n, Manuel Lainz, Asier L\'{o}pez-Gord\'{o}n, and Juan~Carlos
  Marrero.
\newblock A new perspective on nonholonomic brackets and hamilton-jacobi
  theory.
\newblock {\em arXiv:2307.06049v1 [math-ph]}, 2023.

\bibitem{LeMaMa}
Manuel de~Le\'{o}n, Juan~C. Marrero, and Eduardo Mart\'{\i}nez.
\newblock Lagrangian submanifolds and dynamics on {L}ie algebroids.
\newblock {\em J. Phys. A}, 38(24):R241--R308, 2005.

\bibitem{DragovicGajicJova2023}
Vladimir Dragovi\'{c}, Borislav Gaji\'{c}, and Bo\v{z}idar Jovanovi\'{c}.
\newblock Gyroscopic {C}haplygin systems and integrable magnetic flows on
  spheres.
\newblock {\em J. Nonlinear Sci.}, 33(3):Paper No. 43, 51, 2023.

\bibitem{Eden51}
R.~J. Eden.
\newblock The {H}amiltonian dynamics of non-holonomic systems.
\newblock {\em Proc. Roy. Soc. London Ser. A}, 205:564--583, 1951.

\bibitem{Ehlers}
Kurt Ehlers, Jair Koiller, Richard Montgomery, and Pedro~M. Rios.
\newblock Nonholonomic systems via moving frames: {C}artan equivalence and
  {C}haplygin {H}amiltonization.
\newblock In {\em The breadth of symplectic and {P}oisson geometry}, volume 232
  of {\em Progr. Math.}, pages 75--120. Birkh\"{a}user Boston, Boston, MA,
  2005.

\bibitem{FassoGiaSan}
F.~Fass\`o, A.~Giacobbe, and N.~Sansonetto.
\newblock Periodic flows, rank-two {P}oisson structures, and nonholonomic
  mechanics.
\newblock {\em Regul. Chaotic Dyn.}, 10(3):267--284, 2005.

\bibitem{FaGaSa}
Francesco Fass\`o, Luis~C. Garc\'{\i}a-Naranjo, and Nicola Sansonetto.
\newblock Moving energies as first integrals of nonholonomic systems with
  affine constraints.
\newblock {\em Nonlinearity}, 31(3):755--782, 2018.

\bibitem{FassoSanso}
Francesco Fass\`o and Nicola Sansonetto.
\newblock Conservation of `moving' energy in nonholonomic systems with affine
  constraints and integrability of spheres on rotating surfaces.
\newblock {\em J. Nonlinear Sci.}, 26(2):519--544, 2016.

\bibitem{FedoJova}
Y.~N. Fedorov and B.~Jovanovi\'{c}.
\newblock Nonholonomic {LR} systems as generalized {C}haplygin systems with an
  invariant measure and flows on homogeneous spaces.
\newblock {\em J. Nonlinear Sci.}, 14(4):341--381, 2004.

\bibitem{FeGaMa}
Yuri~N. Fedorov, Luis~C. Garc\'{\i}a-Naranjo, and Juan~C. Marrero.
\newblock Unimodularity and preservation of volumes in nonholonomic mechanics.
\newblock {\em J. Nonlinear Sci.}, 25(1):203--246, 2015.

\bibitem{GaBa}
Sneha Gajbhiye and Ravi~N. Banavar.
\newblock Geometric modeling and local controllability of a spherical mobile
  robot actuated by an internal pendulum.
\newblock {\em Internat. J. Robust Nonlinear Control}, 26(11):2436--2454, 2016.

\bibitem{GaTo}
Jorge~S. Garcia and Tomoki Ohsawa.
\newblock Controlled lagrangians and stabilization of euler--poincaré
  equations with symmetry breaking nonholonomic constraints.
\newblock {\em arXiv:2303.13086 [math.OC]}, 2023.

\bibitem{GN07}
L.~Garc\'{\i}a-Naranjo.
\newblock Reduction of almost {P}oisson brackets for nonholonomic systems on
  {L}ie groups.
\newblock {\em Regul. Chaotic Dyn.}, 12(4):365--388, 2007.

\bibitem{LGN-thesis}
L.~C. Garcia-Naranjo.
\newblock {\em Almost {P}oisson brackets for nonholonomic systems on {L}ie
  groups}.
\newblock ProQuest LLC, Ann Arbor, MI, 2007.
\newblock Thesis (Ph.D.)--The University of Arizona.

\bibitem{GN10}
Luis Garc\'{\i}a-Naranjo.
\newblock Reduction of almost {P}oisson brackets and {H}amiltonization of the
  {C}haplygin sphere.
\newblock {\em Discrete Contin. Dyn. Syst. Ser. S}, 3(1):37--60, 2010.

\bibitem{Naranjo}
Luis~C Garc\'ia-Naranjo.
\newblock Hamiltonisation, measure preservation and first integrals of the
  multi-dimensional rubber routh sphere.
\newblock {\em Theoretical and Applied Mechanics}, 46:65--88, 2019.

\bibitem{Futurepaper}
Luis~C. Garc\'ia-Naranjo, Juan~Carlos Marrero, David Mart\'in~de Diego, and
  E.~Petit-Villarreal.
\newblock Geometry of nonholonomic mechanical systems with affine constraints.
\newblock 2023.

\bibitem{GNMo18}
Luis~C. Garc\'{\i}a-Naranjo and James Montaldi.
\newblock Gauge momenta as {C}asimir functions of nonholonomic systems.
\newblock {\em Arch. Ration. Mech. Anal.}, 228(2):563--602, 2018.

\bibitem{GLMM}
J.~Grabowski, M.~de~Le\'{o}n, J.~C. Marrero, and D.~Mart\'{\i}n~de Diego.
\newblock Nonholonomic constraints: a new viewpoint.
\newblock {\em J. Math. Phys.}, 50(1):013520, 17, 2009.

\bibitem{ILMM99}
Alberto Ibort, Manuel de~Leon, Juan~C. Marrero, and David Martin~de Diego.
\newblock Dirac brackets in constrained dynamics.
\newblock {\em Fortschr. Phys.}, 47(5):459--492, 1999.

\bibitem{Jova}
Bo\v{z}idar Jovanovi\'{c}.
\newblock Hamiltonization and integrability of the {C}haplygin sphere in {$\Bbb
  R^n$}.
\newblock {\em J. Nonlinear Sci.}, 20(5):569--593, 2010.

\bibitem{Koiller92}
Jair Koiller.
\newblock Reduction of some classical nonholonomic systems with symmetry.
\newblock {\em Arch. Rational Mech. Anal.}, 118(2):113--148, 1992.

\bibitem{KoonMarsden}
Wang~Sang Koon and Jerrold~E. Marsden.
\newblock Poisson reduction for nonholonomic mechanical systems with symmetry.
\newblock volume~42, pages 101--134. 1998.
\newblock Pacific Institute of Mathematical Sciences Workshop on Nonholonomic
  Constraints in Dynamics (Calgary, AB, 1997).

\bibitem{Krupkova}
Olga Krupkov\'{a}.
\newblock {\em The geometry of ordinary variational equations}, volume 1678 of
  {\em Lecture Notes in Mathematics}.
\newblock Springer-Verlag, Berlin, 1997.

\bibitem{MacPrz-GyrostaticSuslov}
A.~J. Maciejewski and M.~Przybylska.
\newblock Gyrostatic {S}uslov problem.
\newblock {\em Russ. J. Nonlinear Dyn.}, 18(4):609--627, 2022.

\bibitem{Markeev}
A.P. Markeev.
\newblock On integrability of problem on rolling of ball with multiply
  connected cavity filled by ideal liquid.
\newblock {\em Proc. USSR Acad. Sci., Rigid Body Mech.}, 21(2):64--65, 1985.

\bibitem{Marle98}
Charles-Michel Marle.
\newblock Various approaches to conservative and nonconservative nonholonomic
  systems.
\newblock volume~42, pages 211--229. 1998.
\newblock Pacific Institute of Mathematical Sciences Workshop on Nonholonomic
  Constraints in Dynamics (Calgary, AB, 1997).

\bibitem{MaRa1}
Jerrold~E. Marsden and Tudor~S. Ratiu.
\newblock {\em Introduction to mechanics and symmetry}, volume~17 of {\em Texts
  in Applied Mathematics}.
\newblock Springer-Verlag, New York, second edition, 1999.
\newblock A basic exposition of classical mechanical systems.

\bibitem{Neimark}
Ju.~I. Ne\u{\i}mark and N.~A. Fufaev.
\newblock {\em Dynamics of nonholonomic systems}, volume~33 of {\em
  Translations of Mathematical Monographs}.
\newblock American Mathematical Society, Providence, RI, 1972.
\newblock Translasted from the 1967 Russian original by J. R. Barbour.

\bibitem{Ohsawa}
Tomoki Ohsawa.
\newblock Geometric kinematic control of a spherical rolling robot.
\newblock {\em J. Nonlinear Sci.}, 30(1):67--91, 2020.

\bibitem{Petit-thesis}
E.~Petit.
\newblock {\em {Geometry, Dynamics and Mechanics of Nonholonomic Systems with
  Affine Nature}}.
\newblock PhD thesis, Departament of Mathematics. University of Trento, May
  2023.

\bibitem{Ramos}
Arturo Ramos.
\newblock Poisson structures for reduced non-holonomic systems.
\newblock {\em J. Phys. A}, 37(17):4821--4842, 2004.

\bibitem{MaschkevanderSchaft}
A.~J. van~der Schaft and B.~M. Maschke.
\newblock On the {H}amiltonian formulation of nonholonomic mechanical systems.
\newblock {\em Rep. Math. Phys.}, 34(2):225--233, 1994.

\bibitem{SeveraWeinstein}
Pavol \v{S}evera and Alan Weinstein.
\newblock Poisson geometry with a 3-form background.
\newblock Number 144, pages 145--154. 2001.
\newblock Noncommutative geometry and string theory (Yokohama, 2001).

\bibitem{WaKr}
L.-S. Wang and P.~S. Krishnaprasad.
\newblock Gyroscopic control and stabilization.
\newblock {\em J. Nonlinear Sci.}, 2(4):367--415, 1992.

\bibitem{Weber}
R.~W. Weber.
\newblock On {H}amiltonian systems with nonholonomic constraints.
\newblock In {\em Proceedings of the {IUTAM}-{ISIMM} symposium on modern
  developments in analytical mechanics, {V}ol. {I} ({T}orino, 1982)}, volume
  117, pages 457--461, 1983.

\bibitem{Weber2}
Ren\'{e}~W. Weber.
\newblock Hamiltonian systems with constraints and their meaning in mechanics.
\newblock {\em Arch. Rational Mech. Anal.}, 91(4):309--335, 1986.

\end{thebibliography}
\end{document}